\documentclass[10pt,a4paper]{article}

\usepackage{fullpage}
\usepackage{microtype}
\usepackage{libertine}

\usepackage{enumerate}
\usepackage{hyperref}
\usepackage{amsmath,amssymb,amsthm}
\usepackage{tikz}
\usepackage{algorithm}
\usepackage{algpseudocode}
\usepackage{fullpage}
\usepackage{todonotes}
\usepackage{authblk}
\usepackage{xspace}
\usepackage[normalem]{ulem}

\usetikzlibrary{calc}

\usepackage[absolute]{textpos}
\usepackage[all=normal,bibliography=tight]{savetrees}

\newtheorem{theorem}{Theorem}[section]
\newtheorem{lemma}[theorem]{Lemma}
\newtheorem{corollary}[theorem]{Corollary}
\newtheorem{proposition}[theorem]{Proposition}

\newtheorem{claim}{Claim}

\theoremstyle{definition}
\newtheorem{definition}[theorem]{Definition}

\theoremstyle{remark}

\newcommand\abs[1]{\lvert #1\rvert}

\newcommand{\Oh}{\mathcal{O}}

\newcommand{\flow}{\mathcal{P}}
\newcommand{\witnessflow}{\widehat{\flow}}

\newcommand{\block}{\mathbf{B}}

\newcommand{\grantthankstext}{This research is a part of a project that have received funding from the European Research Council (ERC) under the European Union's Horizon 2020 research and innovation programme
Grant Agreement 714704 (M. Pilipczuk). Eun Jung Kim is supported by the grant from French National Research Agency under JCJC program (ASSK: ANR-18-CE40-0025-01).}
\newcommand{\grantthanks}{\parbox[t]{0.78\linewidth}{\grantthankstext{}}\ \parbox[t]{0.18\linewidth}{~\\[-3mm]\includegraphics[height=30px]{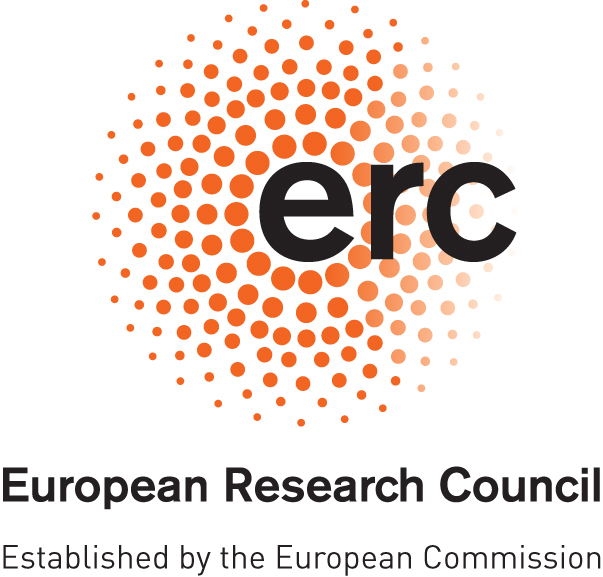}\ \includegraphics[height=30px]{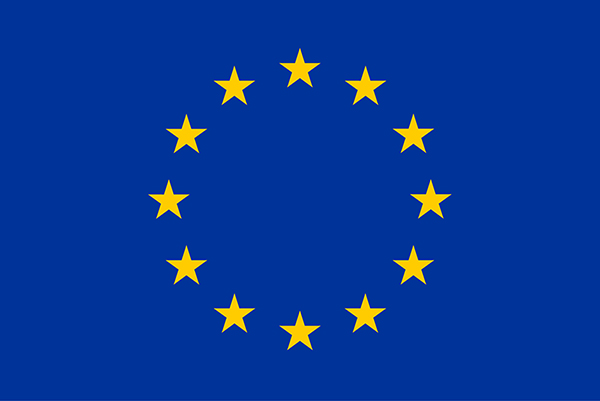}}}

\title{Flow-augmentation II: Undirected graphs%
\thanks{\grantthanks{}\\A preliminary version of this work was presented at SODA 2021~\cite{ufl-soda}.}}
\date{}

\author[1]{Eun Jung Kim}
\author[2]{Stefan Kratsch}
\author[3]{Marcin Pilipczuk}
\author[4]{Magnus Wahlstr\"{o}m}
\affil[1]{Universit\'{e} Paris-Dauphine, PSL Research University, CNRS, UMR 7243, LAMSADE, 75016, Paris, France.}
\affil[2]{Humboldt-Universit\"at zu Berlin, Germany}
\affil[3]{University of Warsaw, Warsaw, Poland}
\affil[4]{Royal Holloway, University of London, TW20 0EX, UK}

\begin{document}

\begin{titlepage}
\def\thepage{}
\thispagestyle{empty}

\maketitle

\begin{abstract}
We present an undirected version of the recently introduced \emph{flow-augmentation} technique:
Given an undirected multigraph $G$ with distinguished vertices $s,t \in V(G)$ and an integer $k$,
one can in randomized $k^{\Oh(1)} \cdot (|V(G)| + |E(G)|)$ time sample a set $A \subseteq \binom{V(G)}{2}$ such that the following holds: for every inclusion-wise minimal $st$-cut
$Z$ in $G$ of cardinality at most $k$, $Z$ becomes a \emph{minimum-cardinality} cut
between $s$ and $t$ in $G+A$ (i.e., in the multigraph $G$ with all edges of $A$ added)
  with probability $2^{-\Oh(k \log k)}$.

Compared to the version for directed graphs [STOC 2022], 
the version presented here has improved success probability ($2^{-\Oh(k \log k)}$ instead of $2^{-\Oh(k^4 \log k)}$),
linear dependency on the graph size in the running time bound,
and an arguably simpler proof. 

An immediate corollary is that the \textsc{Bi-objective $st$-Cut} problem
can be solved in randomized FPT time $2^{\Oh(k \log k)} (|V(G)|+|E(G)|)$ on undirected graphs.
\end{abstract}

\end{titlepage}

\section{Introduction}

Fixed-parameter tractable algorithms for graph separation problems has been an important
question in parameterized complexity, and after more than a decade of intense study
it would seem that we should by now know of all the major techniques necessary for
the design of such algorithms.  Certainly, there is an impressive toolbox, leading to
the resolution of central problems such as FPT algorithms for 
\textsc{Multicut}~\cite{MarxR14,BousquetDT18}
and \textsc{Minimum Bisection}~\cite{CyganLPPS19}.

Yet despite this progress, several open problems remained until very recently.
Many of these relate to directed graph cuts, such as the
existence of FPT algorithms for the notorious \textsc{$\ell$-Chain SAT}
problem identified by Chitnis et al.~\cite{ChitnisEM17},
weighted variants of classic problems such as \textsc{Directed Feedback Vertex Set},
or the deceptively simple-looking problem of \textsc{Bi-objective $(s,t)$-cut}~\cite{KratschLMPW20}.
In the last problem, the input is a digraph  $D=(V,A)$ with arc weights $w$ and $s, t \in V$, and
two budgets $k, W$, and the task is to find an $(s,t)$-cut $Z \subseteq A$ such that
$|Z| \leq k$ and $w(Z) \leq W$.  Despite the simplicity of the problem,
the existence of an FPT algorithm was open for a long time.

This paper is a second one in a series that introduces a new algorithmic technique
of \emph{flow-augmentation} and explores its applications. 
The first part~\cite{dfl-arxiv} introduced the technique in full generality in directed
graphs and applied it to show fixed-parameter tractability of 
\textsc{$\ell$-Chain SAT} and weighted \textsc{Directed Feedback Vertex Set}. 
The third one~\cite{csp-arxiv} uses the technique to show two new tractability isles
in the area of parameterized algorithms for
Constraint Satisfaction Problems, parameterized by the number of unsatisfied clauses,
and completes a complexity dichotomy in the Boolean domain. 
The main goal of this part is to show a counterpart in undirected graphs
that is simpler and has improved guarantees as compared to the (more general) directed
version of part one~\cite{dfl-arxiv}.

To state formally the main result, we some notation.
Consider an undirected graph $G=(V,E)$ with two vertices $s, t \in V$,
and an unknown $(s,t)$-cut $Z$. Furthermore, let $Z_{s,t} \subseteq Z$
be those edges with one endpoint reachable from  $s$ and the other reachable from $t$ in $G-Z$.
We say that $Z$ is a \emph{special $(s,t)$-cut} if $Z_{s,t}$ is an $(s,t)$-cut,
and \emph{eligible for $(s,t)$}  if additionally every edge of $Z$ has its endpoints
in different connected components of $G-Z$.
In particular, any minimal, not necessarily minimum $(s,t)$-cut is
eligible for $(s,t)$.
Let $k=|Z|$, $\lambda^\ast=|Z_{s,t}|$, and let $\lambda_G(s,t) \leq \lambda^\ast$
be the value of an $(s,t)$-max flow in $G$.  
We show the following 
(reformulated slightly from the more formal version
 in Section~\ref{section:specialcutsaugmentation}).

\begin{theorem} \label{theorem-intro:flow-augmentation}
  There is a randomized algorithm that, given an
  undirected graph $G=(V,E)$ with $s, t \in V$ and two
  integers $k \geq \lambda^\ast \geq \lambda_G(s,t)$,
  in time $k^{\Oh(1)}(|V|+|E|)$
  outputs an edge multiset $A$ with $\lambda_{G+A}(s,t) \geq \lambda^*$
  and a flow $\witnessflow$ in $G+A$ of cardinality $\lambda^\ast$,
  such that for any $(s,t)$-cut $Z$ in $G$ eligible for $(s,t)$ with $|Z|=k$ and $|Z_{s,t}|=\lambda^*$, 
  with probability $2^{-\Oh(k \log k)}$, the following holds:
  for every $uv \in A$, $u$ and $v$ are connected in $G-Z$; and 
  for every path $P \in \witnessflow$, $|E(P) \cap Z|=1$.
\end{theorem}
In particular, in any successful run, in $G+A$ the paths $\witnessflow$ will be an $(s,t)$ max-flow,
and $Z_{s,t}$ will be an $(s,t)$-min cut.

A quick comparison of 
Theorem~\ref{theorem-intro:flow-augmentation} with the directed version of~\cite{dfl-arxiv} is in order.
\begin{itemize}
\item There is a better success probability bound: $2^{-\Oh(k \log k)}$ instead of $2^{-\Oh(k^4 \log k)}$.
\item There is an explicit linear dependency on the graph size in the running time bound,
  instead of just a polynomial of unspeficied degree of~\cite{dfl-arxiv}.
\item The notion of an eligible cut is a bit more general than 
the natural casting of the notion of \emph{star $st$-cut} of~\cite{dfl-arxiv}
to undirected graphs (it allows some part of $Z \setminus Z_{s,t}$ to separate a bunch of vertices from $t$, even though these vertices are already separated from $s$ by $Z_{s,t}$).
\item The algorithm and the proof is arguably simpler than the one of~\cite{dfl-arxiv}
(albeit it involves a good amount of tedious calculations in the probability analysis
 to reach the $2^{-\Oh(k \log k)}$ bound).
\item While in directed graphs we provided a deterministic counterpart
with the expected $2^{\Oh(k^4 \log k)}$ parametric factor in the running time bound,
we do not present an analogous result here.
All random steps in the presented algorithm can be replaced by with branching or
standard derandomization tools for color-coding,
so obtaining some deterministic counterpart is definitely possible.
However, to achieve $2^{-\Oh(k \log k)}$ success probability
we needed to carefully optimise probability distributions in a few places
and it is not clear to us that the standard derandomization would match the desired $2^{\Oh(k \log k)}$ parametric factor in the running time bound.
Furthermore, the determinization tools for color-coding steps will introduce a number of $\Oh(\log n)$ factors in the running time analysis, 
  turning the linear dependency on the graph size into a near-linear one $(|V|+|E|)^{1+o(1)}$. 
We remark also that for any complexity classification results (such those in~\cite{csp-arxiv}), the deterministic 
version of the more general directed case of~\cite{dfl-arxiv} is sufficient. 
Finally, presenting the deterministic counterpart along the randomized proof would significantly cloud the picture. 
\end{itemize}

Recall the \textsc{Bi-objective $(s,t)$-Cut} problem. 
Papadimitriou and Yannakakis showed that this is strongly NP-hard,
even for undirected graphs, and also showed partial approximation
hardness~\cite{PapadimitriouY01}. The directed version, with $\ell
\geq 2$ distinct budgets, was recently considered from a parameterized
perspective by Kratsch et al.~\cite{KratschLMPW20}, who showed
that the problem is FPT if all budgets are included in the parameter, 
but W[1]-hard if at least two budgets $k_i$ are not included in the
parameter. The case of a single budget not being included in the
parameter, which includes the \textsc{Bi-objective $(s,t)$-Cut} problem parameterized by $k$,
 has been open prior to our work (in directed graphs).

If $k$ equals the minimum cardinality of an $(s,t)$-cut, the problem can be easily
solved via any polynomial-time minimum cut algorithm: set the capacity of every edge to be a large
number (much larger than any weight of an edge) plus the weight of an edge and ask for a minimum capacity cut. 
Hence, flow-augmentation yields a simple randomized FPT algorithm:
We prepend the step above with flow augmentation (Theorem~\ref{theorem-intro:flow-augmentation}
    in undirected graphs and the version of~\cite{dfl-arxiv} in directed graphs),  
with newly added edges assigned prohibitively large weights.
For undirected graphs, this gives the following corollary.
\begin{corollary}\label{cor:bi-cut}
\textsc{Bi-objective $(s,t)$-Cut} in undirected graphs can be solved in randomized FPT time
$2^{\Oh(k \log k)} \cdot (|V(G)|+|E(G)|)$.
\end{corollary}
We remark that although it is a quite standard exercise to provide an FPT algorithm
for \textsc{Bi-objective $(s,t)$-Cut} in \emph{undirected graphs} within the framework
of randomized contractions~\cite{ChitnisCHPP16},
and recent improvements would also give $2^{\Oh(k \log k)}$ parametric factor~\cite{CyganKLPPSW21},
these techniques do not give any explicit bound on the polynomial factor in the running time bound,
not to mention guaranteeing a linear one.

\section{Preliminaries}\label{section:preliminaries}

In this work we consider only (finite) undirected \emph{multi-graphs without loops}. In particular, different edges connecting the same pair of vertices are considered to be identifiable and non-interchangeable.\footnote{This generality seems necessary to cover a largest set of applications. Multiple copies of the same edge in $G$ might arise in the reduction of some problem to an appropriate cut problem. The different copies may have wildly different behavior regarding contribution to solution cost. Our goal will be to ensure that all solutions of a certain cardinality in terms of cut size have a good probability of being preserved, thereby remaining oblivious to many unnecessary details of the application.} Formally, a \emph{multi-graph} could be captured as $G=(V,E,\pi)$ where $V$ and $E$ are finite sets and $\pi\colon E\to\binom{V}{2}$ assigns each edge in $E$ an unordered pair of endpoints. To keep notation within reason, we will treat multi-graphs as pairs $G=(V,E)$ where $V$ is a finite set and $E$ is a multi-subset of $\binom{V}{2}$ but understanding that cuts $X$ (to be defined in a moment) could involve deleting particular (identifiable) copies of virtually the same edge $uv$.
For a multi-graph $G$ and $A$ a multi-set of edges on $V$, the graphs $G+A$ and $G-A$ are accordingly understood as starting from $G$ and, respectively, adding all edges in $A$ that are not yet in $G$ or removing from $G$ all edges that are also in $A$; again, note that this may include different edges with the same two endpoints.
For a vertex set $S$, we denote by $\delta(S)$ the multi-set of edges that have precisely one endpoint in $S$, and by $\partial(S)$ the set of vertices in $S$ that are incident with at least one edge in $\delta(S)$.
By a \emph{connected component} we mean a maximal set $S\subseteq V$ that induces a connected subgraph of $G$.
In all other aspects we follow standard graph notation as set out by Diestel~\cite{Diestel_book}.

Throughout this paragraph let $G=(V,E)$ be an arbitrary multi-graph, let $S,T\subseteq V$, and let $X\subseteq E$. Define $R_S(X)$ as the set of vertices that are reachable from any vertex in $S$ in $G-X$. The set $X$ is an \emph{$(S,T)$-cut} if $R_S(X)\cap R_T(X)=\emptyset$; note that no such cut exists if $S\cap T\neq\emptyset$. A \emph{minimum $(S,T)$-cut} is any $(S,T)$-cut of minimum possible cardinality; whereas $X$ is a \emph{minimal $(S,T)$-cut} if no proper subset of $X$ is an $(S,T)$-cut. (We will crucially need both minimum and minimal cuts.) By the well-known duality of cuts and flows in graphs (Menger's theorem suffices here) we get that the cardinality of any minimum $(S,T)$-cut is equal to the maximum number of edge-disjoint paths from $S$ to $T$ in $G$ or, equivalently, to the maximum unit-capacity $(S,T)$-flow. 
By $\lambda_G(S,T)$ we denote the maximum flow from $S$ to $T$ or, equivalently, the minimum size of an $(S,T)$-cut in $G$; we omit the subscript $G$ when it is clear from context. We mostly apply these notions for the special cases of $S=\{s\}$ and $T=\{t\}$ and then write, e.g., $(s,t)$-cut rather than $(\{s\},\{t\})$-cut for succinctness. In particular, we write $\lambda_G(s,t)$ rather than $\lambda_G(\{s\},\{t\})$ and, when $G$, $s$, and $t$ are understood, we usually abbreviate this to $\lambda$.
We say that an $(S,T)$-cut $X$ is \emph{closest to $S$} if for every other $(S,T)$-cut $X'$ with $R_S(X')\subseteq R_S(X)$ we have $|X'|>|X|$. Clearly, if $X$ is an $(S,T)$-cut closest to $S$ then $X$ must in particular be minimal.

Let us recall two useful facts about edge cuts in graphs.

\begin{proposition}\label{proposition:reachable}
Let $X$ be a minimal $(S,T)$-cut. Then $X=\delta(R_S(X))=\delta(R_T(X))$.
\end{proposition}

\begin{proof}
By definition of $R_S(X)$ we must have $\delta(R_S(X))\subseteq X$. As $X$ is an $(S,T)$-cut, we have $T\cap R_S(X)=\emptyset$ and, thus, $\delta(R_S(X))$ is also an $(S,T)$-cut. Minimality of $X$ now implies that $X=\delta(R_S(X))$; the other equation works symmetrically.
\end{proof}

\begin{proposition}\label{proposition:uniqueclosestmincut}
There is a unique minimum $(S,T)$-cut that is closest to $S$.
\end{proposition}

\begin{proof}
We use the well-known fact that the cut function $f\colon 2^V\to\mathbb{N}\colon Z\mapsto\abs{\delta(Z)}$ is submodular. Suppose that there are two different minimum $(S,T)$-cuts $X$ and $Y$ that are both closest to $S$. We must have $R_S(X)\neq R_S(Y)$ or else $X=\delta(R_S(X))=\delta(R_S(Y))=Y$ by Proposition~\ref{proposition:reachable} as $X$ and $Y$ must be minimal $(S,T)$-cuts. Using submodularity of $f$ for the sets $R_S(X)$ and $R_S(Y)$ we get
\begin{align}
\abs{\delta(R_S(X))} + \abs{\delta(R_S(Y))} \geq \abs{\delta(R_S(X)\cap R_S(Y))} +\abs{\delta(R_S(X)\cup  R_S(Y))}.\label{math:submod}
\end{align}
Clearly, both $\delta(R_S(X)\cap R_S(Y))$ and $\delta(R_S(X)\cup  R_S(Y))$ are $(S,T)$-cuts and by (\ref{math:submod}) they must both be minimum $(S,T)$-cuts. Now, however, because $R_S(X)\neq R_S(Y)$ we must have $R_S(X)\cap R_S(Y)\subsetneq R_S(X)$ or $R_S(X)\cap R_S(Y)\subsetneq R_S(Y)$ contradicting the assumption that $X$ and $Y$ are both closest to $S$.
\end{proof}

The simple argument used in the proof of Proposition~\ref{proposition:uniqueclosestmincut} is called the \emph{uncrossing of minimum cuts}: From the minimum cuts $X=\delta(R_S(X))$ and $Y=\delta(R_S(Y))$ we obtain minimum cuts $\delta(R_S(X)\cap R_S(Y))$ and $\delta(R_S(X)\cup R_S(Y))$. While the reachable sets $R_S(X)$ and $R_S(Y)$ are in general incomparable, it is clear that $R_S(X)\cap R_S(Y)\subseteq R_S(X)\cup R_S(Y)$, and equality can only hold if $\delta(R_S(X))=\delta(R_S(Y))$.

\section{Undirected flow-augmentation}%
\label{section:specialcutsaugmentation}
\newcommand{\probfun}{g}
\newcommand \probpair {\textsc{Cutting edge-pair rechability}}
\newcommand \probCSPnegcut {\textsc{Almost $\Pi$-CSP}}
\newcommand \probsample {\textsc{flow-augmentation sampling}\xspace}
\newcommand \algosample{\textsc{Sample}\xspace}
\newcommand \algoshort{\textsc{Short-separation}\xspace}
\newcommand \algoshortone{\textsc{Short-separation-single}\xspace}
\newcommand \paths{\textsc{Witnessing-flow}\xspace}
\newcommand \pathshort{\textsc{Witnessing-inner}\xspace}
\newcommand\lefti{{\sf left}}
\newcommand\righti{{\sf right}}
\newcommand\leftcut{{\sf leftcut}}
\newcommand\rightcut{{\sf rightcut}}
\newcommand{\W}{\mathcal{W}}

\paragraph{Special cuts, eligible cuts, compatibility, and flow augmentation.}
Let $G=(V,E)$ be a connected, undirected multi-graph, and let vertices $s,t\in V$. 
For $Z\subseteq E$, let $Z_{s,t}\subseteq Z$ be the set of edges with one endpoint in $R_s(Z)$ and one endpoint in $R_t(Z)$.

The following notions are crucial for this section.
\begin{definition}[special cut]
We say that an $(s,t)$-cut $Z$ is \emph{special} if
$Z_{s,t}$ is an $(s,t)$-cut.
That is, the set of edges $Z_{s,t}\subseteq Z$ with one endpoint in $R_s(Z)$ and one endpoint in $R_t(Z)$ is also an $(s,t)$-cut.
\end{definition}
Note that special $(s,t)$-cuts generalize minimal $(s,t)$-cuts. 

In this section, we focus on solutions that are special $(s,t)$-cuts with an additional technical property.
\begin{definition}[eligible cut]
We say that an $(s,t)$-cut $Z$ is \emph{eligible for $(s,t)$} if
\newcounter{eligibleprops}
 \begin{enumerate}
 \item $Z$ is special, and\label{probsample:special}
  \item each edge of $Z$ has its endpoints in different connected components of $G-Z$. \label{probsample:minimal}
  \setcounter{eligibleprops}{\value{enumi}}
 \end{enumerate}
For an integer $\lambda^*$, we say that an $(s,t)$-cut $Z$ is $\lambda^*$-eligible if $Z$ is eligible and 
additionally $|Z_{s,t}| = \lambda^*$.
\end{definition}

The next two definitions formalize two properties we want from a set of edges that we add to the graph: (i) it does not break the solution, and (ii) it increases the flow from $s$ to $t$.
\begin{definition}[compatible set]
A multi-subset $A$ of $\binom{V}{2}$ is \emph{compatible} with a set $Z \subseteq E$ if for every $uv \in A$, $u$ and $v$ are connected in $G-Z$.
\end{definition}
\begin{definition}[flow-augmenting set]
For an integer $\lambda^\ast \geq \lambda_G(s,t)$, a multi-subset $A$ of $\binom{V}{2}$ is \emph{$\lambda^\ast$-flow-augmenting}
if $\lambda_{G+A}(s,t) \geq \lambda^\ast$.
\end{definition}

Intuitively, the role of $Z$ will be played by an unknown solution to the cut problem in question and compatibility of $A$ with $Z$ means that $A$ cannot add connectivity that was removed by $Z$ (or that was not present in the first place). The challenge is to find a flow-augmenting set that with good probability is consistent with at least one solution $Z$, without knowing $Z$ beforehand.

It will be convenient to take edges in $A$ as being \emph{undeletable} or, equivalently, as unbounded (or infinite) capacity. Clearly, if $A$ is flow-augmenting and compatible with an (eligible) set $Z$ then $A$ remains flow-augmenting and compatible with $Z$ after adding an arbitrary number of copies of any edges in $A$. In particular,
having a total of $k+1$ copies of every edge in $A$ will make those edges effectively undeletable for sets $Z$ of size $k$, that is, the endpoints of any edge in $A$ cannot be separated by $Z$. Note that for applications, since edges in $A$ are in addition to the original input, one will usually not be interested in deleting edges of $A$ anyway (and costs may not be defined),
and they only help to increase the flow to match an (unknown) solution. For the purpose of flow and path packings, edges in $A$ may, accordingly, be shared by any number of (flow) paths, fully equivalent to simply having $k+1$ copies of each edge.

\paragraph{Witnessing flow.}
Similarly as in the directed case, in addition to returning a flow-augmenting set, we will also attempt to return an $(s,t)$-max flow in the augmented graph which intersects $Z_{s,t}$ in a particularly structured way.

In the following, let $G$ be a connected graph with $s, t \in V(G)$, and let $Z$ be an $(s,t)$-cut in $G$ which contains an $(s,t)$-min cut. 
A \emph{witnessing $(s,t)$-flow for $Z$} in $G$ is an $(s,t)$-max flow $\witnessflow$ in $G$ such that every edge of $Z_{s,t}$ occurs on a path of $\witnessflow$, and every path of $\witnessflow$ intersects $Z$  in precisely one edge. 

We make a few observations.  First, since $Z$ is an $(s,t)$-cut, every $(s,t)$-path in $G$ intersects $Z$ in at least one edge. 
Second, if additionally $\lambda_{G}(s,t)=|Z_{s,t}|$, then every $(s,t)$-max flow in $G$ is witnessing \emph{for $Z_{s,t}$}. 
Hence, if $Z$ is a minimum $(s,t)$-cut, then finding a witnessing flow is no harder than finding a flow-augmenting set.
However, if $Z$ is a special and only $Z_{s,t}$ is a minimum $(s,t)$-cut, then a witnessing flow is a more restrictive notion.

We now observe that for every special $(s,t)$-cut $Z$, one can augment $G$ with a set compatible with $Z$ such that $Z_{s,t}$ becomes a $(s,t)$-min cut and 
$G+A$ admits a witnessing flow for $G$.

\begin{lemma}\label{lemma:allexist}
Let $G=(V,E)$ be a multi-graph, let $s,t\in V$ with $s\neq t$, let $Z\subseteq E$ be a special $(s,t)$-cut of size $k$, and let $\lambda^\ast = |Z_{s,t}|$.
Then there exists a $\lambda^\ast$-flow-augmenting set $A$ compatible with $Z$ and a witnessing flow $\witnessflow$ for $Z$ in $G+A$.
\end{lemma}
\begin{proof}
For each pair $u$ and $v$ of vertices in the same connected component of $G-Z$, add to $A$ a set of $k+1$ copies of the edge $uv$.
Clearly, $A$ is compatible with $Z$.
For every $e = uv \in Z_{s,t}$ with $u \in R_s(Z)$ and $v \in R_t(Z)$, let $P_e$ be a path in $G+A$ consisting of the edges $su \in A$, $uv \in Z_{s,t}$, and $vt \in A$.
Then, $\witnessflow := \{P_e~|~e \in Z_{s,t}\}$ is a witnessing flow for $Z$ in $G+A$ of cardinality $\lambda^\ast$. Hence, $A$ is $\lambda^\ast$-flow-augmenting.
\end{proof}
A few remarks are in place. 
The proof of Lemma~\ref{lemma:allexist} shows that a set $Z \subseteq E$ admits a $\lambda^\ast$-flow-augmenting set $A$ if and only if $Z$ does not contain
an $(s,t)$-cut of cardinality less than $\lambda^\ast$. Indeed, in one direction such a cut $C \subseteq Z$ remains an $(s,t)$-cut in $G+A$, preventing the flow from increasing above $|C|$,
and in the other direction the set $A$ constructed in the proof of Lemma~\ref{lemma:allexist} is in some sense ``maximum possible'' and all $(s,t)$-cuts of cardinality at most $k$ in $G+A$
are contained in $Z$. 
Furthermore, even if $Z$ is a special $(s,t)$-cut where $Z_{s,t}$ is an $(s,t)$-min cut (so no flow increase is possible), 
while $Z$ may not admit a witnessing flow in $G$, it is possible to augment $G$ with a set of edges compatible with $Z$ so that a witnessing flow exists.

Lemma~\ref{lemma:allexist} motivates the following extension of the definition of compatibility.
\begin{definition}[compatible pair]
A pair $(A,\witnessflow)$ is \emph{compatible} with a special $(s,t)$-cut $Z$ if
$A$ is a $\lambda^\ast$-flow-augmenting set compatible with $Z$ for $\lambda^\ast = |Z_{s,t}|$
and $\witnessflow$ is a witnessing flow for $Z$ in $G+A$.
\end{definition}

\paragraph{Problem formulation.}
The proof of Lemma~\ref{lemma:allexist} shows that the task of finding a compatible flow-augmenting set and a witnessing flow would be trivial if only we knew $Z$ in advance.
Not knowing $Z$, we will have to place additional edges more sparingly than in the proof of Lemma~\ref{lemma:allexist} to arrive at a sufficient success probability.
Let us formally define our goal, taking into account that the set $Z$ is not known.

In the \probsample problem we are given an instance $(G,s,t,k, \lambda^*)$ consisting of an undirected multi-graph $G=(V,E)$, vertices $s,t\in V$, 
and integers $k$ and $\lambda^*$ such that $k\geq \lambda^* \geq \lambda:= \lambda_G(s,t)$. 
The goal is to find (in probabilistic polynomial-time) a multi-set  $A$ of $V \choose 2$ and an $(s,t)$-flow $\witnessflow$ in $G+A$ such that the following holds:
 \begin{itemize}
 \item $\lambda_{G+A}(s,t)\geq\lambda^*$, $|\witnessflow| = \lambda^\ast$, and
 \item for each $\lambda^*$-eligible $(s,t)$-cut $Z$ of size exactly $k$, the output $(A,\witnessflow)$ is compatible with $Z$ with probability at least $p$.
 \end{itemize}
The function $p$ (that may depend on $k$ or $\lambda$) is called the \emph{success probability}.

In order to relax some corner cases, we allow for the event that $\lambda_{G+A}(s,t)>\lambda^*$, and note that if $Z$ is an eligible $(s,t)$-cut with $|Z_{s,t}|=\lambda^*$ then for any such output $(A,\witnessflow)$ such that $A$ is compatible with $Z$ we must have $\lambda_{G+A}(s,t)=\lambda^*$. 

\paragraph{Results.}
We can now formulate the main result of this section.  

\begin{theorem} \label{theorem:flow-augmentation}
There is a randomized polynomial-time algorithm that, given a \probsample{} instance $(G,s,t,k,\lambda^*)$ with $\lambda_G(s,t) \leq \lambda^* \leq k$,
outputs a set $A$ with $\lambda_{G+A}(s,t) \geq \lambda^*$ and a flow $\witnessflow$ in $G+A$ of cardinality $\lambda^\ast$,
such that the following holds: for each set $Z\subseteq E$ of size $k$ that is $\lambda^*$-eligible for $(G,s,t,k)$, the output $A$ is compatible with $Z$ and $\witnessflow$ is a witnessing
flow for $Z$ in $G+A$ with success probability $2^{-\Oh(k \log k)}$.
The algorithm can be implemented to run in time $k^{\Oh(1)}\Oh(m)$. 
\end{theorem}

The rest of the section is devoted to proving Theorem~\ref{theorem:flow-augmentation}.
For clarity, we will only argume polynomial-time running time bound through the proof,
and only discuss how to reach $k^{\Oh(1)} \Oh(m)$ bound in the end.

We begin by introducing an appropriate decomposition of (the vertex set of) $G$ into what we call \emph{bundles}, which in turn consist of what is called \emph{blocks}. We then present our recursive flow-augmentation algorithm, splitting the presentation into an ``outer loop'' and an ``inner loop.'' Note that we assume that the input multi-graph $G$ is connected as this somewhat simplifies presentation, but we will circumvent this assumption in applications.

It will be convenient to assume that we only care about $\lambda^*$-eligible cuts that do not contain any edge incident with $s$ nor $t$. 
This can be easily achieved by adding an extra terminal $s'$ connected with $s$ with $k+1$ edges, adding an extra terminal $t'$ connected with $t$ 
with $k+1$ edges, and asking for $(s',t')$-cuts instead. 
Consequently, in the proof we can assume one more property of an $\lambda^*$-eligible $(s,t)$-cut $Z$:
\begin{enumerate}
\setcounter{enumi}{\value{eligibleprops}}
\item $Z$ contains no edge incident with $s$ or $t$. \label{probsample:propertyST} 
\end{enumerate}

\subsection{Blocks and bundles}\label{subsection:blocksandbundles}

Given an instance $(G,s,t,k,\lambda^*)$ of \probsample, it should come as no surprise that the minimum $(s,t)$-cuts of $G$ will be crucial for flow augmentation. Recall, however, that even structurally simple graphs may exhibit an exponential number of possibly crossing minimum $(s,t)$-cuts. We will use the notion of closest cuts (and implicitly the well-known uncrossing of minimum $(s,t)$-cuts as used in Proposition~\ref{proposition:uniqueclosestmincut}) to identify a sequence of non-crossing minimum $(s,t)$-cuts. The parts between consecutive cuts will be called blocks; we will also define a partition of blocks into consecutive groups called bundles. The decomposition of $G$ into bundles will guide the choice of edges for the flow-augmenting set $A$ in our algorithm and will be used to capture parts of $G$ to recurse on.

For convenience, let us fix an instance $(G,s,t,k,\lambda^*)$ and let $\lambda:=\lambda_G(s,t) \leq k$ for use in this subsection. Accordingly, in $G$ there is a packing of $\lambda$ edge-disjoint $(s,t)$-paths $P_1,\ldots,P_\lambda$ (and no larger packing exists). Clearly, every minimum $(s,t)$-cut in $G$ contains exactly one edge from each path $P_j$ and no further edges. As noted earlier, we assume for now that $G$ is connected.

\paragraph{Blocks.}
We first define a sequence $C_0,\ldots,C_p$ of non-crossing minimum $(s,t)$-cuts; recall that minimum $(s,t)$-cuts in $G$ all have cardinality $\lambda$. To start, let $C_0$ be the unique minimum $(s,t)$-cut that is closest to $s$. Inductively, for $i\geq 1$, let $C_i$ be the minimum $(s,t)$-cut closest to $s$ among all cuts that fulfil $N[R_s(C_{i-1})]\subseteq R_s(C_i)$. The cut $C_i$ is well-defined (i.e., unique) by an easy variant of Proposition~\ref{proposition:uniqueclosestmincut}: Minimum cuts $X$ fulfilling the requirement that $N[R_s(C_{i-1})]\subseteq R_s(X)$ uncross into minimum cuts fulfilling the same requirement. Intuitively, the construction is equivalent to asking that each $C_i$ is closest to $s$ among minimum $(s,t)$-cuts that do not intersect $C_0\cup\ldots\cup C_{i-1}$ but this would need a formal proof and we do not require it.

We can now define the \emph{blocks} $V_0,\ldots,V_{p+1}\subseteq V$, which will be seen to form a partition of $V$. Block $V_0$ is simply set to $R_s(C_0)$. For $i\in\{1,\ldots,p\}$, we define block $V_i$ as the set of vertices reachable from $s$ in $G-C_i$ but not in $G-C_{i-1}$, i.e., $V_i:=R_s(C_i)\setminus R_s(C_{i-1})$. Finally, $V_{p+1}$ contains all vertices reachable from $s$ in $G$ but not in $G-C_p$ which, since $G$ is connected, equates to $V_{p+1}=V\setminus R_s(C_p)$. By construction of the cuts $C_i$ we clearly have $s\in R_s(C_0)\subsetneq R_s(C_1)\subsetneq\ldots\subsetneq R_s(C_p)\subseteq V\setminus\{t\}$, so the blocks $V_i$ are all nonempty and clearly form a partition of $V$; see Figure~\ref{fig:bundles}.

\begin{figure}[tb]
\begin{center}
\includegraphics{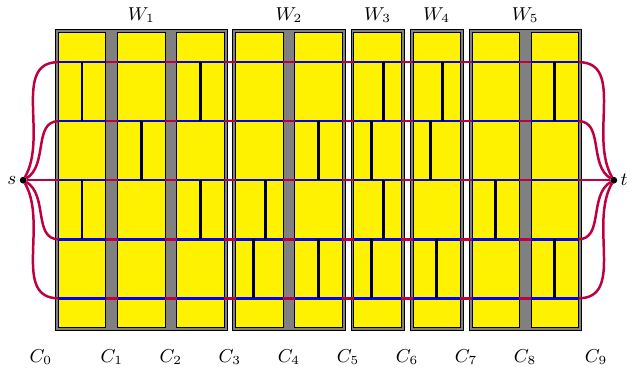}
\caption{A schematic picture of nine blocks (yellow) partitioned into five bundles (gray). The extremal blocks and bundles, containing $s$ and $t$, are not depicted.
  The $s-t$ flow of value $5$ is depicted using blue flow paths, with purple edges being the edges of the consecutive cuts $C_i$.
    The bundle $W_3$ is a connected bundle. }\label{fig:bundles}
\end{center}
\end{figure}

Let us point out that blocks $V_i$ do not need to be connected even though $G$ is connected. It will be useful to note, however, that blocks $V_0$ and $V_{p+1}$ are connected: The graph $G$ is connected and each minimum $(s,t)$-cut $C_i$ will therefore separate it into exactly two connected components $R_s(C_i)$ and $R_t(C_i)$. Blocks $V_0=R_s(C_0)$ and $V_{p+1}=V\setminus R_s(C_p)=R_t(C_p)$ are therefore connected. 
Moreover, each block is at least somewhat connected through subpaths of the flow paths $P_1,\ldots,P_\lambda$ that are contained therein. We establish a bit more structure via the following two propositions.

\begin{proposition}\label{proposition:structureofmincuts}
For each $(s,t)$-flow path $P_j\in\{P_1,\ldots,P_\lambda\}$, seen as being directed from $s$ to $t$, the edges of the minimum $(s,t)$-cuts $C_0,\ldots,C_p$ appear in order of the cuts. These edges define a partition of the flow path $P_j$ into $P^0_j,\ldots,P^{p+1}_j$ so that $P^i_j$ is contained in block $V_i$ for $i\in\{0,\ldots,p+1\}$.
\end{proposition}

\begin{proof}
Fix an $(s,t)$-flow path $P_j$ and let $e_i$ denote the unique edge of $P_j$ that is contained in the minimum $(s,t)$-cut $C_i$, for all $i\in\{0,\ldots,p\}$. Moreover, let $u_i$ and $v_i$ be the endpoints of $e_i$ in order of appearance on $P_j$; note that $v_i=u_{i+1}$ is possible. (We tacitly assume that the flow paths are cycle free.)

Clearly, the vertices of the subpath $P^0_j$ of $P_j$ from $s$ to $u_0$ are contained in $R_s(C_0)$ because no further edge of $P_j$ is in $C_0$. Thus, all vertices of $P_j$ up to and including $v_0$ are contained in $N[R_s(C_0)]\subseteq R_s(C_1)$. Accordingly, the edge $e_1\in C_1$ on $P_j$ must be part of the subpath from $v_0$ to $t$ or else, combined with $v_0\in N[R_s(C_0)]\subseteq R_s(C_1)$ there would be a path from $s$ to $t$. Thus, $e_1$ appears after $e_0$ on $P_j$, when seeing $P_j$ as directed from $s$ to $t$. Observe that $e_0$ and $e_1$ together define a subpath $P^1_j$ from $v_0$ to $u_1$ on $P_j$ that is contained in $R_s(C_1)\setminus R_s(C_0)=V_1$. Iterating this argument for increasing $i$ completes the proof. (Note that the final subpath of $P_j$, denoted $P^{p+1}_j$ starts with $v_p$ and ends in $t$. Clearly, it is contained in $V\setminus R_s(C_p)$.)
\end{proof}

Using the fact that, for each $(s,t)$-flow path $P_j$, the blocks $V_i$ contain consecutive subpaths of $P_j$, we can prove that each block has at most $\lambda$ connected components. Moreover, each such component in a block $V_i$, with $i\in\{1,\ldots,p\}$ is incident with some number of edges of $C_{i-1}$ and the same number of edges in $C_i$.

\begin{proposition}\label{proposition:connectivityofblocks}
Each block $V_i$ has at most $\lambda$ connected components. Moreover, each connected component in a block $V_i$, with $i\in\{1,\ldots,p\}$, is incident with $c\geq 1$ edges in $C_{i-1}$ and with exactly $c$ edges in $C_i$. (Clearly, $V_0$ is incident with all $\lambda$ edges of $C_0$, and $V_{p+1}$ is incident with all $\lambda$ edges of $C_p$.)
\end{proposition}

\begin{proof}
We already know that $V_0$ and $V_{p+1}$ are connected. Consider now a block $V_i$ with $i\in\{1,\ldots,p\}$. Clearly, $R_s(C_i)$ is connected in $G-C_i$ because all of its vertices are reachable from $s$. At the same time, the vertices in $V_i\subseteq R_s(C_i)$ are not reachable from $s$ in $G-C_{i-1}$ by definition of $V_i=R_s(C_i)\setminus R_s(C_{i-1})$, so paths from $s$ to $V_i$ must use at least one edge in $C_{i-1}$. Thus, each connected component $K$ of $V_i$ must be incident with at least one edge of $C_{i-1}$; let $c\geq 1$ be the number of such edges. Now, recall that the edges in $C_{i-1}$ together with those in $C_i$ define the subpaths of $P_1,\ldots,P_\lambda$ that are in block $V_i$. This implies that the $c$ edges of $C_{i-1}$ that are incident with component $K$ in block $V_i$ correspond to exactly $c$ subpaths of paths $P_j\in\{P_1,\ldots,P_\lambda\}$ that are part of component $K$. This of course implies that $K$ must be incident by the $c$ edges of $C_i$ that define those paths. Since connected components of $V_i$ do not share vertices and edges of $C_{i-1}$ and $C_i$ have exactly one endpoint in $V_i$ each, no two connected components of block $V_i$ can share their incident edges in $C_{i-1}$ or $C_i$. Thus, there are at most $\lambda$ connected components in each block $V_i$. This completes the proof.
\end{proof}

It can be easily verified that the decomposition into blocks can be computed in polynomial time.

\begin{proposition}\label{proposition:computingblocks}
Given a multi-graph $G=(V,E)$ and vertices $s,t\in V$, the unique sequence of cuts $C_0,\ldots,C_p$ and decomposition of blocks $V_0,\ldots,V_{p+1}$ can be computed in polynomial time.
\end{proposition}

\begin{proof}
This comes down to computing a polynomial number of closest minimum cuts. A closest minimum $(S,T)$-cut $C$ can be computed by a standard maximum (unit capacity) flow algorithm based on maintaining a residual graph: When the maximum $(S,T)$-flow is reached, let $R\supseteq S$ be the set of vertices that are reachable from $S$ in the residual graph. Clearly, when viewing packed paths as being directed from $S$ to $T$, there is no path edge entering $R$ from $V\setminus R$ because that would yield an edge leaving $R$ in the residual graph. As the flow is maximum, $T$ may not be reachable in the residual graph, so $R\cap T=\emptyset$. Consequently, each path in the packing leaves $R$ exactly once and does not return. Thus, the cardinality of $C$ is equal to the number of paths in the packing, say $\lambda$, making it a minimum $(S,T)$-cut.

To see that $C$ is closest to $S$, assume that there was a minimum $(S,T)$-cut $C'\neq C$ with $R_S(C')\subseteq R_S(C)$. Since both $C$ and $C'$ are also minimal cuts, $R_S(C')=R_S(C)$ would imply $C'=C$ by Proposition~\ref{proposition:reachable}, so assume that $R_S(C')\subsetneq R_S(C)$ and let $v\in R_S(C)\setminus R_S(C')$. Since the cardinality of $C'$ is equal to the size of the path packing, all of its edges are used by $(S,T)$-paths leaving $R_S(C')$. Thus, in the residual graph, there is no edge leaving $R_S(C')$ and hence no path from $S\subseteq R_S(C')$ to $v\notin R_S(C')$; a contradiction.
\end{proof}

\paragraph{Bundles.}
We will now inductively define a decomposition of $V$ into \emph{bundles} $W_0,\ldots,W_{q+1}$; see also Figure~\ref{fig:bundles}. The first bundle $W_0$ is simply equal to the (connected) block $V_0$, which contains $s$. For $i\geq 1$, supposing that blocks $V_0,\ldots,V_{j-1}$ are already parts of previous bundles, 
\begin{itemize}
 \item let $W_i:=V_j$ if $V_j$ is connected (i.e., if $G[V_j]$ is connected) and call it a \emph{connected bundle}
 \item otherwise, let $W_i:=V_j\cup\ldots\cup V_{j'}$ be the union of contiguous blocks, where $j'$ is maximal such that $G[V_j\cup\ldots\cup V_{j'}]$ is \emph{not} connected and call it a \emph{disconnected bundle}.
\end{itemize}
Observe that the final bundle is $W_{q+1}=V_{p+1}$ because $V_{p+1}$ is connected and, due to the included subpaths of $(s,t)$-flow paths (cf.\ Proposition~\ref{proposition:connectivityofblocks}), any union $V_j\cup\ldots\cup V_{p+1}$ induces a connected graph (see also Proposition~\ref{proposition:connectivityofbundles}). We use $\block(W_i)$ to denote the set of blocks whose union is equal to $W_i$, i.e., $\block(W_i)=\{V_j\}$ and $\block(W_i)=\{V_j,\ldots,V_{j'}\}$ respectively in the two cases above. We say that two bundles $W_i$ and $W_{i'}$ are \emph{consecutive} if $\abs{i-i'}=1$.

Intuitively, bundles are defined as maximal sequences of blocks that permit a good argument to apply recursion in our algorithm.
In case of a single block, if we augment the edges incident with the block, then in the recursive step the cardinality of the maximum flow $\lambda_G(s,t)$ increases. 
In case of a union of contiguous blocks that does not induce a connected subgraph, if we recurse into every connected component independently, we split the budget $k$ in a nontrivial way,
as every connected component contains the appropriate part of at least one flow path of $\mathcal{P}$.

Clearly, the bundles $W_0,\ldots,W_{q+1}$ are well defined and they form a partition of the vertex set $V$ of $G$. We emphasize that $W_0=V_0\ni s$ and $W_{q+1}=V_{p+1}\ni t$ and that they are both connected bundles. We note without proof that the bundles inherit the connectivity properties of blocks because the cuts between blocks combined into a bundle connect their subpaths of $(s,t)$-flow paths $P_1,\ldots,P_\lambda$ into longer subpaths, whereas the incidence to the preceding and succeeding cuts stays the same (see Proposition~\ref{proposition:connectivityofbundles}). For ease of reference, let us denote by $C'_0,\ldots,C'_q$ those cuts among $C_0,\ldots,C_p$ that have endpoints in two different (hence consecutive) bundles, concretely, with $C'_i$ having endpoints in both $W_i$ and $W_{i+1}$; note that $C'_0=C_0$ as $W_0=V_0$ and $C'_q=C_p$ as $W_{q+1}=V_{p+1}$.

\begin{proposition}\label{proposition:connectivityofbundles}
Each bundle $W_i$ has at most $\lambda$ connected components. Moreover, each connected component in a bundle $W_i$, with $i\in\{1,\ldots,q\}$, is incident with $c\geq 1$ edges in $C'_{i-1}$ and with $c$ edges in $C'_i$. (Clearly, $W_0=V_0$ is incident with all $\lambda$ edges of $C'_0=C_0$, and $W_{q+1}=V_{p+1}$ is incident with all $\lambda$ edges of $C'_q=C_p$.)
\end{proposition}

Let us introduce some more notation for bundles: For $0\leq a\leq b\leq q+1$ let $W_{a,b}:=\bigcup_{i=a}^b W_i$. Let $W_{\leq a}:=W_{0,a}$ and $W_{\geq a}:=W_{a,q+1}$. For any (union of consecutive bundles) $W_{a,b}$ we define the \emph{left interface $\lefti(W_{a,b})$} as $\partial(W_{\geq a}) \cap W_{\geq a}$ when $a\geq 1$ and as $\{s\}$ when $a=0$. (I.e., when $a\geq 1$ then $\lefti(W_{a,b})$ are those vertices of $W_{a,b}$ that are incident with the cut $C'_{a-1}$ that precedes bundle $W_a$). Similarly, we define the \emph{right interface $\righti(W_{a,b})$} as $\partial(W_{\leq b}) \cap W_{\leq b}$ when $b\leq q$ and as $\{t\}$ when $b=q+1$. (I.e., when $b\leq q$ then $\righti(W_{a,b})$ are those vertices of $W_{a,b}$ that are incident with the cut $C'_b$ that succeeds bundle $W_b$.) For single bundles $W_i$ the same notation applies using $W_i=W_{i,i}$. A consecutive subsequence of bundles is called a \emph{stretch} of bundles, or simply a stretch. 

While a union of consecutive blocks may be disconnected, this is not true for bundles where, as can be easily checked, any two consecutive bundles together induce a connected subgraph of $G$.

\begin{proposition}\label{proposition:twobundlesconnect}
For any two consecutive bundles $W_i$ and $W_{i+1}$ the graph $G[W_i\cup W_{i+1}]$ is connected.
\end{proposition}

\begin{proof}
If $W_i$ is a connected bundle then, using Proposition~\ref{proposition:connectivityofbundles}, we immediately get that $G[W_i\cup W_{i+1}]$ is connected. If $W_i$ is a disconnected bundle then $G[W_i\cup V_{j'+1}]$ is connected, where $W_i=V_j\cup\ldots\cup V_{j'}$ and $V_{j'+1}$ is the first block of $W_{i+1}$. Using Propositions~\ref{proposition:connectivityofblocks} and~\ref{proposition:connectivityofbundles} we again directly get that adding the remaining blocks of $W_{i+1}$ to $G[W_i\cup V_{j'+1}]$ does not break connectivity.
\end{proof}

Clearly, the decomposition into bundles can be efficiently computed from the one into blocks.

\begin{proposition}\label{proposition:computingbundles}
Given a multi-graph $G=(V,E)$ and vertices $s,t\in V$, the unique sequence of cuts $C'_0,\ldots,C'_q$ and decomposition of bundles $W_0,\ldots,W_{q+1}$ can be computed in polynomial time.
\end{proposition}

\paragraph{Affected and unaffected bundles.}
We will later need to reason about the interaction of a special $(s,t)$-cut $Z\subseteq E$ and $G=(V,E)$ and, hence, about the interaction with the bundles of $G$. We say that a bundle $W$ is \emph{unaffected by $Z$} if $N[W]$ is contained in a single connected component of $G-Z$; otherwise we say that $W$ is \emph{affected by $Z$}. As an example, the cut $Z=C'_i$ affects both $W_i$ and $W_{i+1}$ but no other bundles. Similarly, a cut $Z$ entirely confined to $G[W_i]$ affects only $W_i$, since  $W_{\leq i-1}$ and $W_{\geq i+1}$ are both connected and disjoint from $Z$. The more interesting/difficult cuts $Z$ affect several bundles in a non-trivial way. 

The following observation limits the number and arrangement of affected bundles. It will be important for reducing the general case (probabilistically) to the case where $G$ decomposes into a bounded number of bundles. Concretely, this is the purpose of the outer-loop part of our algorithm, which is presented in the following section.

\begin{lemma}\label{lemma:affectedbundles}
Let $Z\subseteq E$ be an $(s,t)$-cut of size at most $k$. 
Let $0 \leq a \leq b \leq q+1$ and let $\ell$
be the number of indices $a \leq i \leq b$ such that the bundle $W_i$ is affected. 
Then,
  $$\ell \leq 2 \left| W_{a-1,b+1} \cap Z \right|.$$
In particular,  at most $2k$ bundles are affected by $Z$.
\end{lemma}

\begin{proof}
We argue that a bundle $W_i\in\W$ is unaffected if $Z$ contains no edge from $C_i'$ and no edge with both endpoints in $W_{i-1} \cup W_i$ (Condition $\star$). First of all, we have $G[W_{i-1}\cup W_i]-Z=G[W_{i-1}\cup W_i]$ in this case, which is connected by Proposition~\ref{proposition:twobundlesconnect} and all vertices of $W_{i-1}\cup W_i$ are in the same component of $G-Z$. As $Z$ is furthermore disjoint from the min-cut $C'_i$ succeeding $W_i$, all vertices of $N[W_i]$ are in the same connected component. Thus, $W_i$ is not affected. A symmetric argument holds for $W_i$ if $Z$ contains no edge of $C_{i-1}'$ and no edge with both endpoints in $W_i \cup W_{i+1}$ (Condition $\dagger$).

To bound the number of affected bundles, first note that $W_0$ and $W_{q+1}$ are affected only if $Z$ contains an edge of $G[W_0]$ or $C_0$, respectively of $G[W_{q+1}]$ or $C_q'$. 
For any remaining bundle $W_i$, say that $W_i$ is \emph{potentially affected} if neither of the conditions $(\star)$ and $(\dagger)$ in the previous paragraph applies.  A bundle $W_i$, $i \in [q]$ is thus potentially affected if (1) $Z$ contains an edge of $G[W_i]$ or intersects $C_{i-1}' \cup C_i'$, 
or (2) $Z$ contains both an edge of $G[W_{i-1}]$ and of $G[W_{i+1}]$. 
For each potentially affected bundle $W_i$, let us charge some edge of $Z$ as described by 1 point if the former case applies, and otherwise charge one edge of $Z$ in $W_{i-1}$ and in $W_{i+1}$ by 1/2 point each. 
Then the total amount of allocated charge equals the number of potentially affected bundles, and it can be seen that every edge is charged at most 2 points. The lemma follows. 
\end{proof}

\begin{lemma}\label{lemma:onlyoneinterestingcut}
Let $Z\subseteq E$ be an $(s,t)$-cut of size at most $k$. There is at most one maximal stretch $W_{a,b}$ of bundles such that every bundle $W_i$, $a \leq i \leq b$ is affected by $Z$ and such that $W_{a,b}$ contains both a vertex reachable from $s$ and a vertex reachable from $t$ (in $G-Z$). Moreover, all vertices in the left interface of $W_{a,b}$ are reachable from $s$ and all vertices in the right interface are reachable from $t$.
Finally, if $Z$ is a special $(s,t)$-cut then there must be such a stretch.
\end{lemma}

\begin{proof}
If there is no such stretch of bundles then the lemma holds vacuously. Else, let $W_{i,j}$ be a maximal  stretch of consecutive affected bundles such that $W_{i,j}$ contains a vertex $p$ that is reachable from $s$ in $G-Z$ and a vertex $q$ that is reachable from $t$ in $G-Z$. Because the left interface $X$ of the stretch separates the stretch from $s$, at least one vertex of $X$ must be reachable from $s$ in $G-Z$. The preceding bundle, namely $W_{i-1}$, is unaffected (by choosing the stretch as maximal). Since $N[W_{i-1}]\supseteq X$ contains a vertex reachable from $s$ in $G-Z$ it follows that all vertices in $N[W_{i-1}]$ are reachable from $s$ in $G-Z$. Since $Z$ separates $s$ from $t$, no vertex is reachable from both $s$ and $t$ in $G-Z$. In particular, this holds for all vertices in $X$ and, hence, for all bundles preceding $W_i$. By a symmetric argument, all vertices in the right interface of the stretch, say $Y$, are reachable from $t$ but not from $s$.

For the final part of the lemma, consider any flow path $P$ from $s$ to $t$ in $G$. Since $Z_{s,t}\subseteq Z$ is also an $(s,t)$-cut, the path $P$ must contain some edge $e\in Z_{s,t}$. By definition of $Z_{s,t}$, $e$ has one endpoint reachable from $s$ in $G-Z$ and one reachable from $t$ in $G-Z$. Clearly, there is a bundle $W$ with both endpoints of $e$ in $N[W]$ (as the bundles are a partition of the vertex set, some bundle $W$ contains an endpoint of $e$ and, hence, $N[W]$ must contain both endpoints). This bundle $W$ is affected and contains a vertex reachable from $s$ and one reachable from $t$ in $G-Z$, namely the endpoints of $e$.  By definition of stretch, $W$ is contained in precisely one maximal stretch of affected bundles, and this stretch meets the conditions of the lemma.
\end{proof}

The previous lemma says that each special $(s,t)$-cut $Z$ yields exactly one maximal stretch of affected bundles in which it separates $s$ from $t$ (and possibly creates further connected components). We say that $Z$ \emph{strongly affects} that stretch. For all other maximal affected stretches of bundles we say that they are \emph{weakly affected by $Z$}.  Note that a non-special cut such as $C_i' \cup C_j'$ for $j \geq i+3$ may contain no strongly affected stretch.

Let us make some useful observations about bundles not in the strongly affected stretch.

\begin{proposition} \label{prop:all-nonstrong-stretches}
  Let $Z$ be a special $(s,t)$-cut and let $W_{a,b}$ be the unique strongly affected stretch.  Then the following hold.
  \begin{enumerate}
  \item For every $i < a$, if $W_i$ is an unaffected bundle then $W_i \subseteq R_s(Z)$
  \item For every $i > b$, if $W_i$ is an unaffected bundle then $W_i \subseteq R_t(Z)$
  \item If $W_{i,j}$ is a (maximal) weakly affected stretch with $j<a$, then $\lefti(W_i) \cup \righti(W_j) \subseteq R_s(Z)$ and $(j-i+1) \leq 2|Z \cap W_{i,j}|$
  \item If $W_{i,j}$ is a (maximal) weakly affected stretch with $i>b$, then $\lefti(W_i) \cup \righti(W_j) \subseteq R_t(Z)$ and $(j-i+1) \leq 2|Z \cap W_{i,j}|$
  \end{enumerate}
\end{proposition}

\subsection{The outer loop of the algorithm}\label{subsection:outerloop}

Our algorithm \algosample consists of an outer loop (to be explained in this section), which is applied first to an input instance $(G,s,t,k,\lambda^*)$ and also to certain instances in recursive calls, and an inner loop, which is applied only to short sequences of bundles. The outer loop part uses a color-coding approach to guess weakly and strongly affected stretches of bundles in $G$, and calls the inner-loop subroutine called \algoshort on the latter. This subroutine (to be described in detail in the following section) then seeks to recursively find an output $(A,\witnessflow)$,
 using the assumption that whenever it is called on a stretch $W_{a,b}$, then either $Z$ is disjoint from the stretch $W_{a,b}$ or $W_{a,b}$ is precisely the unique strongly affected stretch in $G$. 

Each call to our algorithm will return a pair $(A,\witnessflow)$ for the instance in question, where $(A,\witnessflow)$ may or may not be compatible for an arbitrary (unknown) $(s,t)$-cut $Z$.
A crucial observation for the correctness of our algorithm is that any flow-augmentation set guessed for an unaffected stretch of bundles will always be compatible with $Z$.
This allows us to focus our attention in the analysis on the guesses made while processing affected bundles.  This is essential in bounding the success probability purely in terms of $k$.

We will argue that for some sufficiently large constants $c_1 \gg c_2 \gg 0$, \algosample{}($G,s,t,k,\lambda^*$) returns an output $(A,\witnessflow)$ which is, with probability at least $e^{-\probfun(\lambda_G(s,t), k)}$, 
compatible with an (unknown) eligible $(s,t)$-cut
$Z$, where $\probfun(\lambda, k) = (c_1 k - c_2) (1 + \ln k) + c_2 \max(0, k - \lambda)$. 

The main (outer loop) algorithm is shown in Figure~\ref{fig:alg-outer}.
  
\begin{figure}[ht]
\small
\noindent {\bf Algorithm} \algosample{}($G,s,t,k,\lambda^*$)
\begin{enumerate}
\item \label{step:outer-one} If 
it does not hold that $\lambda_G(s,t) \leq \lambda^\ast \leq k$, then 
set $A$ to be $\max(k+1,\lambda^\ast)$ copies of $\{s,t\}$, $\witnessflow$ to be any
$\lambda^\ast$ of these copies, and return $(A,\witnessflow)$.
\item Initialize $A = \emptyset$ and $\witnessflow$ to be a set of $\lambda^\ast$ zero-length
paths starting in $s$.
\item Compute the partition $V=W_0 \cup \ldots \cup W_{q+1}$ of
  $G$ into bundles. 
\item Go into \textsf{single} mode or \textsf{multiple} mode with probability $1/2$ each.
\begin{itemize}
\item In \textsf{single} mode, set $p_\mathrm{blue} = p_\mathrm{red} = 1/2$.
\item In \textsf{multiple} mode, set $p_\mathrm{blue} = 1/k$, $p_\mathrm{red} = 1-1/k$.
\end{itemize}
\item Randomly color each bundle blue or red; blue with probability $p_\mathrm{blue}$ and red with probability $p_\mathrm{red}$.
\item Randomly sample an integer $\lambda^\ast \leq k' \leq k$ as follows:
set $k' = k$ with probability $1/2$ and with remaining probability sample $\lambda^\ast \leq k' < k$ uniformly at random.
\item For every maximal stretch $W_{a,b}$ of bundles colored with
  the same color, do the following in consecutive order starting with $a=0$, and maintaining
  the property that at the begining of the loop 
  $\witnessflow$ is a family of $\lambda^\ast$ edge-disjoint paths in $G+A$
  starting in $s$ and ending in $\lefti(W_a)$:
  \begin{enumerate}
  \item If $a > 0$, then add to $A$ all edges $uv$ for $u,v \in \righti(W_{a-1}) \cup \lefti(W_a)$;
    (We henceforth refer to the edges added in this step as \emph{link edges}.)
  \item If the stretch is colored red and consists of one bundle in \textsf{single} mode, or at least two and at most $2k'$ bundles in \textsf{multiple} mode, then perform the following:
    \begin{enumerate}
    \item Let $G'$ be the graph $G[N[W_{a,b}]]$ with vertices of
      $W_{\leq a-1}$ contracted to a single vertex $s'$ and
      vertices of $W_{\geq b+1}$ contracted to a single vertex $t'$. If $a=0$, and hence $W_{\leq a-1}=\emptyset$, then instead add a new vertex $s'$ and connect it to $s\in W_a$ via $\lambda$ parallel edges $\{s,s'\}$. Similarly, if $b=q+1$ then $W_{\geq b+1}=\emptyset$ and we instead add a new vertex $t'$ and connect it to $t\in W_b$ via $\lambda$ parallel edges $\{t,t'\}$. Observe that $\deg(s')=\deg(t')=\lambda$.
    \item Do a recursive call:
    \begin{itemize}
    \item In \textsf{single} mode, let $(A',\witnessflow') \gets \algoshortone{}(G', s', t', k',\lambda^*)$.
    \item In \textsf{multiple} mode, let $(A',\witnessflow') \gets \algoshort{}(G', s', t', k',\lambda^*)$.
    \end{itemize}
    \item Update $A$ as follows:
    \begin{itemize}
    \item Add to $A$ all edges of $A'$ that are not incident with $s'$ or $t'$.
    \item For every edge $s'v \in A'$, add to $A$ a separate edge $uv$ for each vertex $u\in \righti(W_{\leq a-1})$. 
          If $a=0$ then ignore edges $s's \in A'$ and for each edge $s'v \in A'$ add $sv$ to $A$;
    \item Analogously, for every edge $vt' \in A'$, add to $A$ a separate edge $vw$
          for each vertex $w\in \lefti(W_{b+1})$.
          If $b=q+1$ then ignore edges $tt' \in A'$ and for each edge $vt' \in A'$ add $vt$ to $A$.
    \end{itemize}
    \item Update $\witnessflow$ as follows:
      For every path $P' \in \witnessflow'$, if the first or last edge of $P'$ belongs to $A'$, replace
      it with one of its corresponding edges in $A$, and then 
      pick a distinct path $P \in \witnessflow$
      and append $P'$ at the end of $P$, using a link edge to connect the endpoints of 
      $P$ and $P'$ if necessary.
    \end{enumerate}
  \item Otherwise:
   \begin{enumerate}
   \item Add to $A$, with multiplicity $k+1$, all edges $\{u,w\}$ with $u \in \righti(W_{a-1})$, taking $u=s$ if $a=0$, and $w \in \lefti(W_{b+1})$, taking $w=t$ if $b=q+1$.
   \item Prolong every path $P \in \witnessflow$ with a link edge (if $a > 0$) and an edge of $A$,
   so that $P$ ends in $\lefti(W_{b+1})$, or in $t$ if $b = q+1$. 
   \end{enumerate}
  \end{enumerate}
\end{enumerate}
\caption{The outer loop algorithm}
\label{fig:alg-outer}
\end{figure}

\paragraph{Interface of the inner loop algorithm.}
The inner-loop algorithm expects as input an instance $(G',s',t',k', \lambda')$ that has two additional properties and will return a pair $(A',\witnessflow')$.
A \emph{valid input $(G',s',t',k',\lambda')$ for the inner loop algorithm} has the following properties:
\begin{enumerate}
 \item The graph $G'$ decomposes into bundles $W'_0,\ldots,W'_{q+1}$, with $1\leq q\leq 2k'$, and such that $W'_0=\{s'\}$ and $W'_{q+1}=\{t'\}$.
 If $q = 1$, then we say that the instance is a \emph{single-bundle} instance, otherwise if $q > 1$ it is a \emph{multiple-bundle} instance. 
 \item We have $\lambda_{G'}(s',t')<\lambda' \leq k'$, i.e., the maximum $(s',t')$-flow in $G'$ is lower than the target flow value $\lambda'$ after augmentation. %
 \end{enumerate}
 Furthermore, let $Z'$ be an $(s',t')$-cut in $G'$.  We say that $Z'$ is a \emph{valid cut for $(G',s',t',k',\lambda')$} if the following hold.
 \begin{enumerate}
 \item $Z'$ is an eligible $(s',t')$-cut in $G'$ with $|Z'| = k$ and $|Z'_{s,t}| = \lambda'$;
 \item $Z'$ affects precisely the bundles $W'_1$, \ldots, $W'_q$ in $G'$
 \end{enumerate} 
In the following section we will describe a realization of this interface by two algorithms called \algoshortone and \algoshort with the following success guarantee:
\begin{itemize}
\item 
for a valid single-bundle instance $(G',s',t',k',\lambda')$, the algorithm \algoshortone returns a flow-augmenting set $A'$ with $\lambda_{G'+A'}(s',t') \geq \lambda'$
and an $(s,t)$-flow $\witnessflow'$ in $G+A$ of size $\lambda'$
such that for every valid cut $Z'$, $(A',\witnessflow')$ is compatible with $Z'$ with probability at least $32 \cdot e^{-\probfun(\lambda_{G'}(s',t'), k')}$;
\item 
for a valid multiple-bundle instance $(G',s',t',k',\lambda')$, the algorithm \algoshort returns a flow-augmenting set $A'$ with $\lambda_{G'+A'}(s',t') \geq \lambda'$
and an $(s,t)$-flow $\witnessflow'$ in $G+A$ of size $\lambda'$
such that for every valid cut $Z'$, $(A',\witnessflow')$ is compatible with $Z'$ with probability at least $32(k')^3 \cdot e^{-\probfun(\lambda_{G'}(s',t'), k')}$. 
\end{itemize}

\paragraph{Correctness of the outer loop part.}
We are now ready to prove correctness of the outer loop algorithm \algosample assuming a correct realization of the inner loop algorithm according to the interface stated above.

It is straightforward to verify the invariant stated in the loop:
at every step, $\witnessflow$ is a family of $\lambda^\ast$ edge-disjoint paths in $G+A$,
starting in $s$ and ending in $\lefti(W_a)$. 
It is also straightforward to verify the feasibility of the updates of $\witnessflow$. 
Furthermore, observe that after the last iteration of the loop, all paths of $\witnessflow$
end in $t$. Thus, at the end of the algorithm $\witnessflow$ is indeed a family of $\lambda^\ast$
edge-disjoint paths from $s$ to $t$ in $G+A$.

We now prove that, in a well-defined sense, most edges in the returned set $A$ are compatible with most minimal $(s,t)$-cuts $Z$.

\begin{lemma}\label{lemma:mostedgesarecompatible}
  Let $W_{a,b}$ be a stretch processed by \algosample such that every bundle of the stretch is unaffected by $Z$. 
  Then every edge added to $A$ while processing $W_{a,b}$ is compatible with $Z$. 
\end{lemma}

\begin{proof}
  Note that every vertex of $N[W_{a,b}]$ is in the same connected component of $G-Z$.
  Hence it suffices to observe that every edge added to $A$ in this phase has both endpoints in $N[W_{a,b}]$. 
\end{proof}

Now we are set to prove correctness of the outer loop algorithm assuming a correct realization of the inner-loop interface.

\begin{lemma}\label{lemma:outerloop:correct}
Assume that an algorithm \algoshort correctly realizes the above interface
such that for every valid single-bundle (multiple-bundle) instance $(G',s',t',k',\lambda')$ with $k' \leq k$, the returned pair $(A',\witnessflow')$ is compatible with a fixed valid cut $Z'$
with probability $32 e^{-\probfun(\lambda_{G'}(s',t'),k')}$ ($32(k')^3 e^{-\probfun(\lambda_{G'}(s',t'),k')}$). 
Then for any $(G,s,t,k,\lambda^*)$, \algosample returns an $(s,t)$-flow-augmenting set $A$ such that $\lambda_{G+A}(s,t)\geq \lambda^*$ and for any eligible $(s,t)$-cut $Z$ in $G$ of size $k$ and with $|Z_{s,t}|=\lambda^*$, the returned pair $(A,\witnessflow)$ is compatible with $Z$ with probability at least $e^{-\probfun(\lambda_{G}(s,t), k)}$.
\end{lemma}

\begin{proof}
  The lemma holds essentially vacuously if \algosample{}$(G,s,t,k,\lambda^*)$ stops at step~\ref{step:outer-one}. 
  Hence we assume $\lambda \leq \lambda^* \leq k$. 
  Since $G$ is connected, $\lambda \geq 1$, hence $k \geq 1$.

We first prove that all calls to \algoshort or \algoshortone are made for valid instances $(G',s',t',k,\lambda^*)$. Let $(G',s',t',k,\lambda^*)$ be an instance on which \algoshort or \algoshortone is called and let $W_{a,b}$ be the stretch that the call corresponds to.  It can be verified that $G'$, relative to minimum $(s',t')$-cuts, decomposes into bundles $\{s'\},W_a,\ldots,W_b,\{t'\}$. A key point here is that $s'$ and $t'$ are both incident with precisely $\lambda$ edges in $G'$, and $\lambda_{G'}(s',t')=\lambda$.  This makes $\delta(s')$ the unique closest minimum $(s',t')$-cut. From this point on, the sequence of closest minimum $(s',t')$-cuts that define blocks and bundles is identical to ones between the blocks that form bundles $W_a,\ldots,W_b$ in $G$. Clearly, $G'[W_a\cup\ldots W_b]\cong G[W_a\cup\ldots\cup W_b]$ (canonically) so we arrive at the same decomposition into bundles. At the end, $\delta(t')$ can be seen to be final closest minimum $(s',t')$-cut that arises when computing blocks and bundles for $(G',s',t')$, using a symmetric argument to the one for $\delta(s')$.

Now, we show the compatibility property. Let $Z$ be any $\lambda^*$-eligible $(s,t)$-cut of size $k$.
By Lemma~\ref{lemma:onlyoneinterestingcut}, there is a unique strongly affected stretch $W_{a,b}$, and by Lemma~\ref{lemma:affectedbundles} at most $2|Z|$ bundles are affected in total.  Let $\ell=b-a+1$ be the number of bundles in $W_{a,b}$ and let $Z' = Z \cap W_{a,b}$. We have $\ell \leq 2|Z'|$ and $\lambda^\ast \leq |Z'| \leq k$.

We are interested in the following success of the random choices made by the algorithm:
the algorithm goes into mode \textsf{single} if $a=b$ and into mode \textsf{multiple} otherwise, 
$k' = |Z'|$, and the coloring of bundles in the loop is such that every bundle of $W_{a,b}$ is red, while $W_{a-1}$, $W_{b+1}$, and every other affected bundle is blue.
Since there are at most $2(k-|Z \cap W_{a,b}|)$ affected bundles that are not in $W_{a,b}$,
the above success happens with probability at least 
\begin{itemize}
\item if $a = b$ and $k = |Z'|$: $2^{-5}$;
\item if $a = b$ and $k > |Z'|$: 
$$2^{-5} (k - \lambda^\ast)^{-1} 2^{-2(k-|Z'|)} \geq 2^{-5} k^{-1} 2^{-2(k-|Z'|)};$$
\item if $a < b$:
\begin{align*}
&(k - \lambda^\ast + 1)^{-1} \cdot k^{- 2 - 2(k-|Z'|)} \cdot (1-1/k)^{\ell}\\
& \quad  \geq k^{-3-2(k-|Z'|) } \cdot (1-1/k)^{2k} \geq 2^{-4} k^{-3-2(k-|Z'|) }.
\end{align*}
\end{itemize}
Henceforth we assume that the above success indeed happens.

If this is the case, then for every two consecutive bundles $W_i$ and $W_{i+1}$
of different colors, either $W_i$ or $W_{i+1}$ is unaffected. In particular, 
all endpoints of the edges of $E(W_i,W_{i+1})$ are in the same connected component of $G-Z$.
Thus, all link edges added to $A$ are compatible with $Z$.

Let us now consider the processing of some maximal monochromatic stretch $W_{c,d}$ other than $W_{a,b}$.  If $W_{c,d}$ is red, then by assumption on the coloring it is a stretch of unaffected bundles, and any edges added are compatible with $Z$ by Lemma~\ref{lemma:mostedgesarecompatible}. Furthermore, any flow $\witnessflow'$ does not intersect $Z$, so the edges appended in the paths of $\witnessflow$ are disjoint with $Z$.

If $W_{c,d}$ is red, then we claim that $\lefti(W_c) \cup \righti(W_d)$ are contained in the same connected component in $G-Z$.  Indeed,
  by assumption on the coloring, any affected bundle in $W_{c,d}$ is contained in some weakly affected stretch $W_{c',d'}$ where the stretch is contained in $W_{c,d}$ in its entirety.  By Prop.~\ref{prop:all-nonstrong-stretches} the endpoints of such a stretch are contained in the same component of $G-Z$, as are the endpoints of any stretch of unaffected bundles. The claim follows. Thus the edges added by \algosample for $W_{c,d}$ are compatible with $Z$. Furthermore, in this case all edges appended to the paths of $\witnessflow$ are from $A$.

Now consider the strongly affected stretch $W_{a,b}$.
Observe that \algosample will make a recursive call to \algoshort or \algoshortone for this stretch; let the resulting instance be $(G',s',t',k',\lambda^*)$.
Note that $Z'$ are the edges of $Z$ contained in $G'$ and that $Z'$ is a valid cut for $(G',s',t',k',\lambda^*)$. Furthermore, $\lambda = \lambda_G(s,t) = \lambda_{G'}(s',t')$.
Indeed, by Lemma~\ref{lemma:onlyoneinterestingcut} $\lefti(W_a) \subseteq R_s(Z)$ and $\righti(W_b) \subseteq R_t(Z)$, and since $W_{a-1}$ (if any) and $W_{b+1}$ (if any) are unaffected, these are entirely contained in $R_s(Z)$ respectively $R_t(Z)$ as well. 
Hence $Z'$ is an eligible $(s',t')$-cut in $G'$. Finally, $|Z'| = k'$ and  $|Z'_{s',t'}|=|Z_{s,t}|=\lambda^*$, and by assumption $Z'$ affects every bundle $W_i$, $1 \leq i \leq q$, of $G'$.

Thus, since \algoshort and \algoshortone implement the inner-loop interface, with probability at least $32(k')^3 e^{-\probfun(\lambda, k')}$ in case of \algoshort
and $32 e^{-\probfun(\lambda, k')}$ in case of \algoshortone, it returns a pair $(A',\witnessflow')$ that is compatible with $Z'$ in $G'$.

We verify that the edges added to $A$ for $A'$ are compatible with $Z$.
The connected components of $G[W_{a-1,b+1}]-Z$ are the same as those of $G'-Z'$ except that the component of $s'$ has $W_{a-1}$ in place of $s'$, and the component of $t'$ contains $W_{b+1}$ instead of $t'$ (respectively, are identical but are missing $s'$ and $t'$ if $a=0$ and/or $b=q+1$). Thus, the only edges in $A$ that could, in principle, be incompatible with $Z$ are those that were added in place of edges in $A'$ that are incident with $s'$ or $t'$. But in all cases, the endpoint replacing $s'$ respectively $t'$ is contained in $R_s(Z)$ respectively $R_t(Z)$, implying that they are compatible with $Z$ in $G$ if they are compatible with $Z'$ in $G'$.

For the family of paths $\witnessflow$, note that if $(A',\witnessflow')$ is compatible with $Z'$,
    then for every $P' \in \witnessflow'$, the path $P'$ intersects $Z'$ in precisely one edge and that edge belongs to $Z'_{s,t} = Z_{s,t}$. Hence, by appending $P'$ to a path $P \in \witnessflow$ we add one
intersection of $P$ with $Z$ and that intersection belongs to $Z_{s,t}$.
Since there is only one strongly affected stretch and in all other cases the edges appended
to the paths of $\witnessflow$ are disjoint with $Z$, $\witnessflow$ is a witnessing flow for $Z$ in $G+A$ as desired. 

Furthermore, the existence of $\witnessflow$ implies that $\lambda_{G+A}(s,t) \geq |\witnessflow| = \lambda^\ast$.

In summary, \algosample produces a pair $(A,\witnessflow)$ that is compatible with $Z$ with probability at least (assuming $c_1 \geq 5$):
  \begin{itemize}
  \item if $a = b$ and $k = |Z'|$: 
  $$2^{-5} \cdot 32 \cdot e^{-\probfun(\lambda, k)} = e^{-\probfun(\lambda, k)};$$
  \item if $a = b$ and $k > |Z'|$: 
  $$2^{-5} k^{-1} 2^{-2(k-k')} e^{-\probfun(\lambda, k')} \geq e^{-5-\ln k -2(k-k')} e^{c_2(k-k')(1+\ln k)} e^{-\probfun(\lambda, k')} \geq e^{-\probfun(\lambda, k')};$$
  \item if $a < b$:
\begin{align*}
& \frac{1}{16} k^{-3-2(k-k') } \cdot 16(k')^3 e^{-\probfun(\lambda, k')} \\
& \quad \geq e^{-\probfun(\lambda, k)} \cdot k^{c_1(k-k')} \cdot (k')^3 \cdot k^{-3-2(k-k')} \\
& \quad \geq e^{-\probfun(\lambda, k)} \cdot k^{(c_1-2)(k-k')} \cdot (k'/k)^3 \\
& \quad \geq e^{-\probfun(\lambda, k)}.
\end{align*}
\end{itemize}
This finishes the proof of the lemma.
\end{proof}

\subsection{Cut splits and the inner loop}\label{subsection:innerloop}

\subsubsection{Single-bundle case}\label{ss:inner-single}

\begin{figure}[ht]
\small
  \noindent {\bf Algorithm} \algoshortone{}$(G,s,t,k,\lambda^*)$
  \begin{enumerate}
  \item \label{step:inner-out} If $(G,s,t,k,\lambda^*)$ is not a valid input,
  or 
it does not hold that $\lambda_G(s,t) \leq \lambda^\ast \leq k$, then 
set $A$ to be $\max(k+1,\lambda^\ast)$ copies of $\{s,t\}$, $\witnessflow$ to be any
$\lambda^\ast$ of these copies, and return $(A,\witnessflow)$.
  \item Let $V=W_0 \cup W_1  \cup W_2$ be the partition of $G$ into bundles.
  \item If $W_1$ is a connected bundle:
    \begin{enumerate}
    \item Let $A_0=\delta(s) \cup \delta(t)$.
    \item Compute $(A,\witnessflow) \gets$\algosample{}$(G+A_0, s, t, k, \lambda^*)$.
    \item Return $(A_0 \cup A, \witnessflow)$. 
    \end{enumerate}
  \item Otherwise: %
    \begin{enumerate}
    \item Let $W_1=W_1^{(1)} \cup \ldots \cup W_1^{(c)}$ be the partition of $G[W_1]$ into connected components, and    
    for each $i \in [c]$ let $\lambda_i$ be the amount of $(s,t)$-flow routed through $W_1^{(c)}$; i.e.,
    $\lambda=\lambda_1+\ldots+\lambda_c$ where $\lambda_i > 0$ for each $i \in [c]$
  \item Randomly sample partitions $\lambda^*=\lambda_1^*+\ldots+\lambda_c^*$ and $k = k_1 + \ldots k_c$ 
  such that $\lambda_i \leq \lambda_i^* \leq k_i$ for each $i \in [c]$.
  \item For every $i \in [c]$,
    let $G^{(i)}=G[W_1^{(i)} \cup \{s,t\}]$ and
    compute $(A_i, \witnessflow_i) \gets $\algosample{}$(G^{(i)},s,t,k_i,\lambda_i^*)$.
  \item Return $(A := \bigcup_{i=1}^c A_i, \witnessflow := \bigcup_{i=1}^c \witnessflow_i)$.
  \end{enumerate}
  \end{enumerate}
  \caption{Inner loop: Algorithm for a single bundle}
  \label{fig:alg-short-one}
\end{figure}

We will now describe an algorithm \algoshortone that realizes the first half of the inner-loop interface from the previous section.
Given a valid single-bundle instance $(G,s,t,k,\lambda^*)$ where $G$ decomposes into bundles $W_0\cup W_1 \cup W_2$, $W_0=\{s\}$ and $W_{2}=\{t\}$,
it will run in (probabilistic) polynomial time and always return
a $\lambda^\ast$-flow augmenting set $A$.
Moreover, for each $(s,t)$-cut $Z$ that is valid for $(G,s,t,k,\lambda^*)$,
the set $A$ is compatible with $Z$ with probability at least $32 e^{-\probfun(\lambda_G(s,t), k)}$.
We call $W_0=\{s\}$ and $W_{2}=\{t\}$ trivial bundles, $W_1$ is the non-trivial bundle. 
The algorithm is given in Figure~\ref{fig:alg-short-one}.

A few remarks are in place. First, if the algorithm exists 
at Step~\ref{step:inner-out}, then no valid cut $Z$ exists and we can deterministically
output a trivially correct answer. 
Second, sampling of values $(\lambda_1^*,\ldots,\lambda_c^*,k_1,\ldots,k_c)$ does not need to be uniform, but we require that each valid output
$(\lambda_1^*,\ldots,\lambda_c^*,k_1,\ldots,k_c)$ is sampled with probability at least $k^{-2c}$.
Note that there are at most $k^{2c}$ valid outputs.
This can be achieved by, e.g., sampling each $\lambda_i^\ast$ and $k_i$ uniformly at random from $\{1,2,\ldots,k\}$
and, if the sampled values do not satisfy the requirements, return one fixed partition instead. 

Let us now analyse the case when $W_1$ is connected.
  
\begin{lemma} \label{lemma:short-connected}
  Let $(G,s,t,k,\lambda^*)$ and $W_1$ be as above, and let $Z'$ be a valid cut for $(G,s,t,k,\lambda^*)$. 
  If $W_1$ is a connected bundle,
  then $\delta(s) \cup \delta(t)$ is a $(\lambda_G(s,t)+1)$-flow-augmenting set compatible with $Z'$.\end{lemma}
\begin{proof}
  Let $A=\delta(s) \cup \delta(t)$.  Since $Z'$ is a valid cut, $Z \cap A = \emptyset$
  and $A$ is compatible with $Z'$.  Furthermore, if $W_1$ is a connected bundle,
  then it consists of a single block. Assume for a contradiction that $G+A$ has an $(s,t)$-cut $C$ of size $\lambda_G(s,t)$.  Then $C \cap A = \emptyset$, and $C$ is an $(s,t)$-min cut in $G$ disjoint from $\delta(s) \cup \delta(t)$.  This contradicts the assumption that $W_1$ a block. Thus every $(s,t)$-min cut in $G$ intersects $\delta(s) \cup \delta(t)$ in at least one edge $e$.  Since $A$ contains a copy of $e$, $C$ is no longer an $(s,t)$-cut in $G+A$.  Hence $G+A$ has no $(s,t)$-cuts of size $\lambda_G(s,t)$, and $\lambda_{G+A}(s,t)>\lambda_G(s,t)$. %
\end{proof}

\begin{lemma} \label{lemma:short-one}
  Assume that \algosample{} is correct for all inputs $(G',s',t',k',\lambda')$
  where either $k'<k$ or $k'=k$ but $\lambda_{G'}(s',t') > \lambda_G(s,t)$,
  with a success probability of at least $e^{-\probfun(\lambda_{G'}(s',t'), k')}$ for any eligible $(s,t)$-cut $Z$.
  Then \algoshortone{}($G,s,t,k,\lambda^*$) is correct, with a success
  probability of at least $32 e^{-\probfun(\lambda_G(s,t), k)}$.
\end{lemma}
\begin{proof}
  Assume that $(G,s,t,k,\lambda^*)$ is a valid input.
  As discussed, we can assume $\lambda_G(s,t) \leq \lambda^\ast \leq k$.
  Let $Z$ be a valid cut. 
  If $W_1$ is a connected bundle, then 
  $A=\delta(s) \cup \delta(t)$ is flow-augmenting and compatible with $Z$ by Lemma~\ref{lemma:short-connected}.
  For the success probability bound, 
  the statement is trivial if $\lambda_{G+A}(s,t) > \lambda^\ast$ (there is no such $Z$
      in this case). Otherwise
  note that $\lambda_{G+A}(s,t) > \lambda_G(s,t)$ so
  $$\probfun(\lambda_G(s,t),k) > \probfun(\lambda_{G+A}(s,t),k) + c_2.$$
  Hence, the probability bound follows as long as $e^{c_2} \geq 32$.
  
  If $W_1$ is a disconnected bundle, let $W_1=W_1^{(1)} \cup \ldots \cup W_1^{(c)}$ be as in the algorithm.
  For $i \in [c]$, let $\lambda_i^*=|Z_{s,t} \cap E(W_1^{(i)})|$ and $k_i = |Z \cap E(W_1^{(i)})|$; then by assumption
  $\lambda^*=\lambda_1^*+\ldots+\lambda_c^*$, $k = k_1 + \ldots k_c$, and $\lambda_i \leq \lambda_i^* \leq k_i$.
  We note that the algorithm guesses the correct values of $k_i$ and $\lambda_i^*$ with probability at least $k^{-2c}$.
  
  Consider some $i \in [c]$ and let $G^{(i)}=G[W_1^{(i)} \cup \{s,t\}]$.
  Let $Z^{(i)}=Z \cap E(G^{(i)})$, and note that
  $Z^{(i)}$ is an $(s,t)$-cut in $G^{(i)}$,
  with endpoints in different connected components of $G^{(i)}-Z^{(i)}$, 
  and with $Z^{(i)} \cap (\delta(s) \cup \delta(t))=\emptyset$.
  Thus $Z^{(i)}$ is eligible for $G^{(i)}$.
  Furthermore by assumption $|Z^{(i)}_{s,t}|=\lambda_i^*$ and $|Z^{(i)}| = k_i < k$.
  Thus each call to \algosample$(G',s,t,k_i,\lambda_i^*)$ will by assumption return a set $A_i$ 
  such that $\lambda_{G+A_i}(s,t) \geq \lambda_i^*$; since $E(G)$ are partitioned across the instances $G^{(i)}$, it follows that $A=A_1 \cup \ldots \cup A_c$ is a flow-augmenting set with $\lambda_{G+A}(s,t)\geq \lambda^*$.  Furthermore, for every $i \in [c]$, with probability at least $e^{-\probfun(\lambda_i, k_i)}$ the set $A_i$ is compatible with $Z^{(i)}$. 
  Now $(A,\witnessflow)$ is compatible with $Z$ if every pair $(A_i,\witnessflow_i)$ is compatible with the respective set $Z^{(i)}$. Hence, the success probability is lower bounded by:
  \begin{align*}
  k^{-2c} \cdot \prod_{i=1}^c e^{-\probfun(\lambda_i, k_i)} & = \exp\left(-2c \ln k - \sum_{i=1}^c \left(c_1(2k_i-\lambda_i) - c_2\right)(1 + \ln k_i)\right) \\
    & \geq \exp\left(-2c \ln k - (1+\ln k) \sum_{i=1}^c c_1(2k_i - \lambda_i) - c_2 \right) \\
    & = \exp\left(-2c \ln k + (1+\ln k) \left(c_1(2k - \lambda) - c_2\right) + (1+\ln k)(c-1)c_2 \right)\\
    & \geq 32 e^{-\probfun(\lambda,k)} \cdot \exp\left(\left((c-1)c_2 - 2c\right) \ln k + \left((c-1)c_2- \ln 32\right) \right) \\
    & \geq 32 e^{-\probfun(\lambda,k)}.
  \end{align*}
  In the above we have used that $c_1 > c_2$ and, in the last inequality, that $c \geq 2$, $c_2 \geq 4 > \ln 32$.
  This finishes the proof of the lemma.
\end{proof}

\subsubsection{Multiple-bundle case}\label{ss:inner-multiple}
We will now describe an algorithm \algoshort that realizes the inner-loop interface from the previous section. Given a valid multiple-bundle instance $(G,s,t,k,\lambda^*)$ where $G$ decomposes into bundles $W_0\cup\ldots\cup W_{q+1}$, with $2\leq q\leq 2k$, and $W_0=\{s\}$ and $W_{q+1}=\{t\}$, and with $\lambda := \lambda_G(s,t)<\lambda^*$ it will run in (probabilistic) polynomial time and always return an $(s,t)$-flow augmenting set $A$. Moreover, for each $(s,t)$-cut $Z$ that is valid for $(G,s,t,k,\lambda^*)$,
the set $A$ is compatible with $Z$ with probability at least $32 k^3 e^{-\probfun(\lambda_G(s,t), k)}$. We call $W_0=\{s\}$ and $W_{q+1}=\{t\}$ trivial bundles; all others are called non-trivial bundles.

\begin{figure}[ht]
\small
\noindent {\bf Algorithm} \algoshort{}($G,s,t,k,\lambda^*$)
\begin{enumerate}
\item If $(G,s,t,k,\lambda^*)$ is not a valid multiple-bundle input, then return $k+1$ copies of the edge $\{s,t\}$ and stop.
\item Let $V=W_0 \cup \ldots \cup W_{q+1}$ be the partition of $G$ into bundles.
Let $C$ be the min-cut between $W_1$ and $W_2$.
  \item Randomly sample values $0 \leq \lambda_\gamma \leq \lambda$ for $\gamma \in \Gamma$ such that $\sum_{\gamma \in \Gamma} \lambda_\gamma = \lambda = |C|$.
  Denote $\lambda_0 = \sum_{\gamma \in \Gamma_0} \lambda_\gamma$.
  \item For every edge $uv \in C$ with $u \in V(W_1)$ and $v \in V(W_2)$, guess a label $\varphi(uv) \in \Gamma$ with the probability of $\varphi(uv) = \gamma$
  being $\lambda_\gamma / \lambda$. 
  \begin{enumerate}
  \item 
  Define $\varphi(u)$ and $\varphi(v)$ such that $(\varphi(u), \varphi(v)) = \varphi(uv)$. 
  If a vertex $x$ obtains two distinct values $\varphi(x)$ in this process, return $A$ being $k+1$ edges $st$ and stop.
  \item Let $\lambda_C^*$ be the number of edges $e \in C$ such that $\varphi(e) \in \{(s,t), (t,s)\}$. 
  \end{enumerate}
  \item Let $A_{st}$  contain $k+1$ copies of each edge $\{u,v\}$ with $u, v \in \{s\} \cup \varphi^{-1}(s)$
    or with $u, v \in \{t\} \cup \varphi^{-1}(t)$
  \item Compute a set $\witnessflow_C$ of size $\lambda_C^\ast$ as follows: for every $e \in C$ such that
  $\varphi(e) \in \{(s,t), (t,s)\}$, let $e = uv$ be such that $\varphi(u) = s$ and $\varphi(v) = t$,and add to $\witnessflow_C$ a three-edge path $P_e$ consisting of the edges $su \in A_{st}$, $e$, and $tv \in A_{st}$. 
  \item Randomly sample a partition $\lambda^*=\lambda_1^*+\lambda_C^*+\lambda_2^*$ subject to the following constraints:
    \begin{enumerate}
    \item $\lambda_1^* \geq \lambda_{(t,s)} + \lambda_{(t,t)} + \lambda_{(t,\bot)}$ and $\lambda_1^* = 0$ 
    if $\lambda_{(t,s)} = \lambda_{(t,t)} = \lambda_{(t,\bot)} = 0$.
    \item $\lambda_2^* \geq \lambda_{(s,s)} + \lambda_{(t,s)} + \lambda_{(\bot,s)}$ and $\lambda_2^* = 0$ 
    if $\lambda_{(s,s)} = \lambda_{(t,s)} = \lambda_{(\bot,s)} = 0$.
    \end{enumerate}
  \item Randomly sample a partition $k = k_1 + k_C + k_2$ subject to the following constraints:
    \begin{enumerate}
    \item $\lambda_1^* \leq k_1$, $\lambda_0 + \lambda_{(t,t)} \leq k_1 + k_C$, $1 \leq k_1$, $\lambda_\leftarrow \leq k_1$;
    \item $\lambda_2^* \leq k_2$, $\lambda_0 + \lambda_{(s,s)} \leq k_2 + k_C$, $1 \leq k_2$, $\lambda_\rightarrow \leq k_2$;
    \item $\sum_{\gamma \in \Gamma_0 \setminus \{(\bot,\bot)\}} \lambda_\gamma \leq k_C \leq \sum_{\gamma \in \Gamma_0} \lambda_\gamma$.
    \end{enumerate}
  \item Construct a flow-augmenting set $A_1$ and a flow in $W_1$:
    \begin{enumerate}
    \item Let $G_1=(G+A_{st})[W_1 \cup \{s,t\}]$;
    \item Compute $(A_1,\witnessflow_1) \gets \algosample(G_1, s,t,k_1,\lambda_1^*)$.
    \end{enumerate}
  \item Construct a flow-augmenting set $A_2$ in $W_{2,q}$:    
    \begin{enumerate}
    \item Let $G_2=(G+A_{st})[W_{2,q} \cup \{s,t\}]$.
    \item Compute $(A_2,\witnessflow_2) \gets \algosample(G_2, s,t,k_2,\lambda_2^*)$.
    \end{enumerate}
  \item Return $(A=A_{st} \cup A_1 \cup A_2, \witnessflow = \witnessflow_C \cup \witnessflow_1 \cup \witnessflow_2)$.
\end{enumerate}
\caption{The inner loop algorithm for multiple-bundle case.}
\label{fig:alg-inner}
\end{figure}

The algorithm is shown in Figure~\ref{fig:alg-inner}, but to discuss it we need
a few results.  Assume that $Z$ is a $\lambda^*$-eligible $(s,t)$-cut which affects every non-trivial bundle $W_1, \ldots, W_q$ of $G$.
Let $C$ be the min-cut between $W_1$ and $W_2$.  We define a \emph{cut labelling}
$\varphi_Z \colon V(C) \to \{s,t,\bot\}$ of $C$ by $Z$ as
\[
  \varphi_Z(v) =
  \begin{cases}
    s & v \in R_s(Z) \\
    t & v \in R_t(Z) \\
    \bot & \text{otherwise.} \\
  \end{cases}
\] 
For every edge $uv \in C$ with $u \in V(W_1)$ and $v \in V(W_2)$, the \emph{type} of the
edge $uv$ is the pair $\varphi_Z(uv) := (\varphi_Z(u), \varphi_Z(v))$. Let $\Gamma = \{s,t,\bot\} \times \{s,t,\bot\}$
be the set of types.
For a type $\gamma \in \Gamma$, let $\lambda_\gamma$ be the number of edges $e \in C$ with $\varphi_Z(e) = \gamma$.
The types $(s,s)$ and $(t,t)$ are somewhat special; we denote $\Gamma_0 = \Gamma \setminus \{(s,s),(t,t)\}$ and $\lambda_0 = \sum_{\gamma \in \Gamma_0} \lambda_\gamma$. 
Furthermore, let $\lambda_{\leftarrow} = \{t, \bot\} \times \{s,t,\bot\}$
and $\lambda_\rightarrow = \{s,t,\bot\} \times \{s, \bot\}$.

Let $Z_1 = Z \cap E(W_1)$, $Z_2 = Z \cap E(W_{2,q})$ and $Z_C = Z \cap C$. 
Note that $Z = Z_1 \cup Z_2 \cup Z_C$ is a partition of $Z$. 
We make some simple observations.
\begin{proposition} \label{prop:multi-recurse-ok}
  The following hold.
  \begin{enumerate}
  \item $Z_{s,t} \cap C = \{e \in C \mid \phi_Z(e) \in \{(s,t), (t,s)\}\}$;
  \item If $Z_{s,t} \cap E(W_1) \neq \emptyset$, then there exists $u \in V(W_1) \cap V(C)$ with $\varphi_Z(u) = t$;
  \item If $Z_{s,t} \cap E(W_{2,q}) \neq \emptyset$, then there exists $v \in V(W_2) \cap V(C)$ with $\varphi_Z(v) = s$;
  \item For every $uv \in C$ such that $\varphi_Z(u) \neq \varphi_Z(v)$, we have $uv \in Z_C$. 
  Conversely, if $uv \in Z_C$, then $\varphi_Z(u) \neq \varphi_Z(v)$ or $\varphi_Z(u) = \varphi_Z(v) = \bot$.
  \item $Z_1 \cup Z_C \neq \emptyset$ and $Z_2 \cup Z_C \neq \emptyset$. 
  \item $|Z_1 \cup Z_C| \geq \lambda_0 + \lambda_{(t,t)}$ and $|Z_2 \cup Z_C| \geq \lambda_0 + \lambda_{(s,s)}$.
  \item $|Z_1| \geq \lambda_\leftarrow$ and $|Z_2| \geq \lambda_\rightarrow$.
  \end{enumerate} 
\end{proposition}
\begin{proof}
  \emph{1.} Holds by definition of $Z_{s,t}$ and $\varphi_Z$.
  
  \emph{2--3.} These proofs are symmetric, so we show only the first.
  Let $\{u,v\} \in Z_{s,t} \cap E(W_1)$ where $v \in R_t(Z)$.  Then there is a path $P$
  from $v$ to $t$ in $G-Z$. This path must pass through $C$ through a vertex $w$
  with $w \in R_t(Z)$. 
  
  \emph{4.} This is straightforward from the assumption that every edge of $Z$ connects two distinct connected components of $G-Z$.

  \emph{5.} If this does not hold, then some non-trivial bundle of $G$ is unaffected.

  \emph{6--7.} Consider a maximum $(s,t)$-flow $(P_e)_{e \in C}$ where $e \in E(P_e)$ for $e = u_1u_2 \in C$, $u_1 \in W_1$, $u_2 in W_2$.
  The path $P_e$ first goes from $s$ via $W_1$ to $e$ and then continues via $W_{2,q}$ to $t$.
  If there no edge of $E(P_e) \cap Z$ between $s$ and $u_1$, then 
  $\varphi_Z(u_1) = s$. IF additionally $e \notin Z$, then $\varphi_Z(e) = (s,s)$.
  This proves the two inequalities of 6--7. concering $Z_1$.   The argument for $Z_2$
  is symmetric.
\end{proof}

We use this to show the correctness of the algorithm. 

\begin{lemma} \label{lemma:short-correct}
  Assume that \algosample{} is correct for all inputs $(G',s',t',k',\lambda')$
  where either $k'<k$ or $k'=k$ but $(k'-\lambda_{G'}(s',t')) < (k-\lambda_G(s,t))$,
  with a success probability of at least $e^{-\probfun(\lambda_{G'}(s',t'), k')}$.
  Then \algoshort{}($G,s,t,k,\lambda^*$) is correct, with a success
  probability of at least $32 k^3 e^{-\probfun(\lambda_{G}(s,t),k)}$.
\end{lemma}
\begin{proof}
  First observe that if a call $(G_i,s,t,k_i,\lambda_i^*)$ is made to \algosample,
  then $s$ and $t$ are connected in $G_i$.  Indeed, $G[W_1 \cup \{s\}]$ is connected,
  and if $\varphi^{-1}(t) \cap V(W_1) = \emptyset$ then the algorithm always
  guesses $\lambda_1^*=0$, hence no recursive call is made.  Similarly,
  $G[W_{2,q} \cup \{t\}]$ is connected and if $s$ is not adjacent
  to $V(W_2)$ in $G+A_0$ then no recursive call into $G_2$ is made.
  Hence each recursive call is only made to a connected graph $G_i$ and we can assume that
  $\witnessflow_i$ is a flow of size $\lambda_i^\ast$ in $G_i+A_i$.
  We show that $\witnessflow$ is a flow of size $\lambda^\ast$ in $G+A$,
  which implies that $\lambda_{G+A}(s,t) \geq \lambda^\ast$.
  Indeed, the paths of $\witnessflow_1 \cup \witnessflow_2$ exist in $G+A$ and are pairwise edge-disjoint. 
  Furthermore, for every edge $e \in C$ with $\varphi(e) \in \{(s,t),(t,s)\}$,
  the constucted path $P_e \in \witnessflow_C$
  is a path from $s$ to $t$ disjoint from $\witnessflow_1 \cup \witnessflow_2$.
  Since $|\witnessflow_C| = \lambda^\ast_C$ and $\lambda^\ast = \lambda_1^\ast + \lambda_C^\ast + \lambda_2^\ast$, $\witnessflow$ is as desired.

  Next, we consider the probability that $(A,\witnessflow)$ is compatible with $Z$.
  The algorithm correctly guesses (in every bullet, we condition on the previous guesses being correct):
  \begin{itemize}
  \item values $\lambda_\gamma$ for $\gamma \in \Gamma$ with probability at least
  $(1+\lambda)^{-|\Gamma|}  \geq k^{-9}$;
  \item $\varphi = \varphi_Z$ with probability
  $$\prod_{e \in C} \frac{\lambda_{\varphi_Z(e)}}{\lambda} = \prod_{\gamma \in \Gamma} \left(\frac{\lambda_\gamma}{\lambda}\right)^{\lambda_\gamma} = \exp\left(-\sum_{\gamma \in \Gamma} \lambda_\gamma \ln(\lambda/\lambda_\gamma)\right).$$
  \item values $\lambda_1^* = |Z_{s,t} \cap E(W_1)|$ and $\lambda_2^* = |Z_{s,t} \cap E(W_{2,q})|$
  with probability at least $k^{-2}$;
  \item values $k_1 = |Z \cap E(W_1)|$, $k_2 = |Z \cap E(W_2)|$, $k_C = |Z \cap C|$
  with probability at least $k^{-2}$, as there are at most $k^2$ possible values
  of $(k_1,k_2)$.
  \end{itemize}
  Proposition~\ref{prop:multi-recurse-ok} ensures that 
  in all of the above guesses, the correct value of is among one of the options
  with positive probability.
  Furthermore, $\lambda_C^* = |Z_{s,t} \cap C|$ is computed (deterministically) 
  by the algorithm.

  It was argued above that each recursive call on a graph $G_i$, $i=1, 2$, is made only if $G_i$ is connected. 
  We claim that furthermore $Z_1:=Z \cap E(W_1)$ is an eligible $(s,t)$-cut in $G_1$.
  Indeed, $Z_1 \cap \delta(s)=\emptyset$ by assumption, and $Z_1 \cap \delta(t)=\emptyset$
  since all edges of $\delta(t)$ in $G_1$ are from $A_0$. Furthermore, by assumption,
  for every vertex $u$ of $N_{G_1}(s)$ and every vertex $v$ of $N_{G_1}(t)$,
  we have $u \in R_s(Z)$ and $v \in R_t(Z)$. Hence $Z_1$ in particular cuts every path from $u$
  to $v$ in $G[W_1]$, and by cutting all these paths $Z_1$ must cut $s$ from $t$ in $G_1$.
  Finally, no edge of $Z_1$ goes within a connected component of $G_1-Z_1$,
  since the only paths that are added to $G[W_1]$ go between vertices of the same component
  (either $R_s(Z)$ or $R_t(Z)$) in $G-Z$. 
  Hence with probability at least $e^{-\probfun(\lambda_{G_1}(s,t), k_1)}$
  (or $1$ if $\lambda_1^\ast = 0$)
  the pair $(A_1,\witnessflow_1)$ is compatible with $Z_1$. 
  All these arguments can also be made symmetrically to argue that 
  with probability at least $e^{-\probfun(\lambda_{G_2}(s,t), k_2)}$
  (or $1$ if $\lambda_2^\ast = 0$),
  $(A_2,\witnessflow_2)$ is compatible with $Z_2$. 
  
  By assumption, $A_{st}$ is compatible with $Z$. Also, if $\varphi=\varphi_Z$,
  then every path $P \in \witnessflow_C$ intersects $Z$ in exactly one edge and this edge belongs
  to $Z_{s,t}$.

  It remains to wrap up the proof of the 
  bound the probability that
  $(A = A_{st} \cup A_1 \cup A_2,\witnessflow = \witnessflow_1 \cup \witnessflow_2 \cup \witnessflow_C)$ is compatible with $Z$. 
  First, consider a corner case when $\lambda_{(s,s)} = \lambda$, that is,
  $\varphi_Z$ is constant at $(s,s)$. Then $k_1 \geq 1$, $k_2 \leq k-1$, $k_C = 0$,
  $\lambda_{G_2}(s,t) \geq \lambda_G(s,t)$,
  and the recursive call on $G_1$ is not made. Furthermore,
  once $\lambda_{(s,s)} = \lambda$ is guessed, $\varphi$ is defined deterministically.
  Hence, for sufficiently large constant $c_1$,
  $(A,\witnessflow)$ is compatible with $Z$ with probability at least
  $$k^{-13} e^{-\probfun(\lambda_{G_2}(s,t), k_2)} \geq k^{-13} e^{-\probfun(\lambda, k-1)} 
  \geq \exp\left(-16 \ln k - \ln 16 + c_2(1+\ln 4)\right) 16k^3 e^{-\probfun(\lambda, k)} \geq e^{-\probfun(\lambda, k)}.$$
  A symmetric argument holds if $\lambda_{(t,t)} = \lambda$, that is, $\varphi_Z$ is constant
  at $(t,t)$.

  For the general case, observe that even if the recursive call on $G_i$
  is not invoked due to $\lambda_i^* = 0$, then $k_i \geq 1$ and $\lambda_{G_i}(s,t) \leq k_i$
  so $e^{-\probfun(\lambda_{G_i}(s,t),k_i)} \leq 1$. Thus, we can use
  $e^{-\probfun(\lambda_{G_i}(s,t),k_i)}$ as a lower bound on the success probability 
  of the recursive call regardless of whether it was actually invoked.

  By the above discussion, the probability that $A$ is compatible with $Z$ is at least
  \begin{equation}\label{eq:fa:1}
  k^{-13} \cdot \exp\left(-\sum_{\gamma \in \Gamma} \lambda_\gamma \ln(\lambda/\lambda_\gamma)\right) \cdot e^{-\probfun(\lambda_{G_1}(s,t), k_1)} e^{-\probfun(\lambda_{G_2}(s,t), k_2)}.
  \end{equation}
  We start by analysing the second term of the above bound. 
  By the concavity of $\ln(\cdot)$, we have that
  \begin{equation}\label{eq:fa:2}
  \sum_{\gamma \in \Gamma_0} \lambda_\gamma \ln \lambda_\gamma \geq \lambda_0 \ln (\lambda_0 / |\Gamma_0) = \lambda_0 \ln \lambda_0 - \lambda_0 \ln 7.
  \end{equation}
  Hence,
  \begin{equation}\label{eq:fa:3}
  \sum_{\gamma \in \Gamma} \lambda_\gamma \ln (\lambda/\lambda_\gamma) \leq
   \lambda_{(s,s)} \ln(\lambda/\lambda_{(s,s)}) + \lambda_{(t,t)} \ln (\lambda/\lambda_{(t,t)})
  + \lambda_0 \ln(\lambda/\lambda_0) + \lambda_0 \ln 7.
  \end{equation}
  Denote $x_1 = k_1 - \lambda_0$, $x_2 = k_2 - \lambda_0$, $x_0 = \lambda_0$, and $x = x_1+x_2+2x_0 = k_1 + k_2$.
  By entropy maximization, 
  \begin{align}
  &\lambda_{(s,s)} \ln(\lambda/\lambda_{(s,s)}) + \lambda_{(t,t)} \ln (\lambda/\lambda_{(t,t)})
  + \lambda_0 \ln(\lambda/\lambda_0) \nonumber\\
    &\quad \leq \lambda_{(s,s)} \ln(x/x_1) + \lambda_{(t,t)} \ln(x/x_2) + \lambda_0 \ln(x/(2x_0)) \label{eq:fa:4}\\
    &\quad \leq x_1 \ln(x/x_1) + x_2 \ln(x/x_2) + 2x_0 \ln(x/(2x_0)).\nonumber
  \end{align}
  We also need the following observation:
  \begin{claim}\label{cl:fa-lambda}
  It holds that 
  $$k_1 - \lambda_{G_1}(s,t) + k_2 - \lambda_{G_2}(s,t) \leq k - \lambda_G(s,t) + \lambda_{(\bot,\bot)}.$$
  \end{claim}
  \begin{proof}
  From Proposition~\ref{prop:multi-recurse-ok}(4.), we infer that 
  $$|C| - k_C - \lambda_{(\bot,\bot)} \leq \lambda_{(s,s)} + \lambda_{(t,t)}.$$
  Since in $G_1$, an endpoint of every edge $e \in C$ with $\varphi_Z(e) = (t,t)$ is connected to $t$ with $k+1$ edges,
  we have 
  $$\lambda_{G_1}(s,t) \geq \lambda_{(t,t)}.$$
  Symmetrically, 
  $$\lambda_{G_2}(s,t) \geq \lambda_{(s,s)}.$$
  As $k_1 + k_2 + k_C = k$ and $\lambda_G(s,t) = |C|$, the claim follows.
  \end{proof}

  To wrap up the analysis, we need the following property of the $z \mapsto z \ln z$ function (for completeness, we provide a proof in Appendix~\ref{app:fa-calculus}):
  \begin{claim}\label{cl:fa-calculus}
  Let $f(z) = z \ln z$ for $z > 0$.
  For every constant $C_1 > 0$ there exists a constant $C_2 > 0$ such that
  for every $x_1,x_2,x_0 > 0$ it holds that
  $$C_2 f(x_1 + x_2 + 2x_0) + f(x_1) + f(x_2) + f(2x_0) \geq f(x_1 + x_2 + 2x_0) + C_2 f(x_1+x_0) + C_2 f(x_2 + x_0) + C_1 x_0.$$
  \end{claim}
  Claim~\ref{cl:fa-calculus} for $C_1 = c_2 + \ln 7$  implies an existence of $C_2 > 0$ (depending on $c_2$) such that
  \begin{equation}\label{eq:fa:5}
  x_1 \ln(x/x_1) + x_2 \ln(x/x_2) + 2x_0 \ln(x/(2x_0)) + x_0 (c_2 + \ln 7) \leq C_2 \left(x \ln x - k_1 \ln k_1 - k_2 \ln k_2\right).
  \end{equation}
  Using the definition of $\probfun(\cdot,\cdot)$, the fact that $\lambda_{(\bot,\bot)} \leq x_0$, and Claim~\ref{cl:fa-lambda}, we obtain that
  \begin{equation}
  \probfun(\lambda_G(s,t), k) \geq 
  \probfun(\lambda_{G_1}(s,t), k_1) + \probfun(\lambda_{G_2}(s,t), k_2) + c_1 \left(x \ln x - k_1 \ln k_1 - k_2 \ln k_2\right) + c_2 (1 + \ln k) - c_2 x_0.
  \label{eq:fa:6}
  \end{equation}
  Thus, we bound the negated exponent of the probability bound of~\eqref{eq:fa:1} as follows:
  \begin{align*}
  & 13 \ln k + \sum_{\gamma \in \Gamma} \lambda_\gamma \ln(\lambda/\lambda_\gamma) + \probfun(\lambda_{G_1}(s,t), k_1) + \probfun(\lambda_{G_2}(s,t), k_2)  & \mathrm{by\ \eqref{eq:fa:3}\ and\ \eqref{eq:fa:4}}\\
  &\quad \leq 13 \ln k + x_1 \ln(x/x_1) + x_2 \ln(x/x_2) + 2x_0 \ln(x/(2x_0)) + x_0 \ln 7 & \mathrm{by\ \eqref{eq:fa:6}}\\
  &\qquad + \probfun(\lambda_{G_1}(s,t), k_1) + \probfun(\lambda_{G_2}(s,t), k_2) &\\
  &\quad \leq 13 \ln k + x_1 \ln(x/x_1) + x_2 \ln(x/x_2) + 2x_0 \ln(x/(2x_0)) + x_0 \ln 7 &  \mathrm{by~\eqref{eq:fa:5},}\\
  &\qquad + \probfun(\lambda_G(s,t), k) + c_2 x_0 - c_2(1+\ln k) & c_2 \geq 16,\ c_1 \geq C_2\\
  &\qquad + c_1 (x_1 + x_0) \ln (x_1+x_0) + c_1 (x_2 + x_0) \ln (x_2+x_0) - c_1 x \ln x &\\
  &\quad \leq \probfun(\lambda_G(s,t), k) - 3 \ln k - \ln 32 &
  \end{align*}
  This finishes the proof of the lemma.
\end{proof}

\subsection{Efficient implementation and final proof}

Finally, we show that the algorithms can be implemented to run in time $k^{\Oh(1)}\Oh(m)$, i.e.,
linear time up to factors of $k$. We first show how to efficiently decompose $G$ into bundles.

\begin{lemma} \label{lemma:fast-bundles}
  Let $G=(V,E)$ be an undirected graph, $s, t \in V$, and let $\lambda^\ast \in \mathbb{Z}$ be given. 
  In time $\Oh(\lambda^\ast m)$ we can either show that $\lambda_G(s,t)>\lambda^\ast$ or 
  compute a max $(s,t)$-flow, blocks, and bundles in $G$.
\end{lemma}
\begin{proof}
  Since all edge capacities are unit we can compute a packing $\witnessflow$ of
  up to $\lambda^\ast+1$ $(s,t)$-paths in time $\Oh(\lambda^\ast m)$
  using Ford-Fulkerson, and if $|\witnessflow|=\lambda^\ast+1$ then we are done.
  Otherwise, assume that $\witnessflow$ is a max-flow of value $|\witnessflow|=\lambda$,
  and let $G'$ be the residual flow graph for $\witnessflow$ on $G$.
  We show how to decompose $G$ into blocks and bundles. 
  Let the \emph{mass} of a vertex set $S$ be $\sum_{v \in S} d(v)$. 

  First, observe that the closest min-cut $C_0$ can be found using a
  simple reachability query in the residual flow graph. Specifically, 
  the first block $V_0$ is precisely the set of vertices reachable 
  from $s$ in $G'$. Hence $V_0$ can be computed in time linear in its
  mass and $C_0=\delta(V_0)$. The sets $V_1$, \ldots, $V_{p+1}$ can be
  computed as follows. Let $i \in [p]$ and let $V'=V_0 \cup \ldots \cup V_{i-1}$. 
  Assume that all sets $V_{i'}$, $i'<i$ have been computed, in total time linear
  in the mass of $V'$. Hence the cut $C_{i-1} = \delta(V')$ is known as well. 
  Contract $V'$ into a single vertex $s'$ and reorient the arcs of $C_{i-1}$
  out from $s'$.  Then $V_i$ is precisely the set of vertices
  reachable from $s'$, and can be computed in time linear in its mass.
  Hence we can decompose $G$ into blocks.

  To further group the blocks into bundles, we only need to be able to
  test connectivity. Recall that the first bundle is just $W_0=V_0$. 
  Assume that we are computing the bundle starting with block $V_a$.
  Label the flow-paths $\witnessflow=\{P_1,\ldots,P_\lambda\}$,
  and initialize a partition $Q$ of $[\lambda]$ corresponding to
  the endpoints of $C_{a-1}$ in $V_a$ (i.e., for every vertex $v \in \lefti(W_a)$
  there is a part $B \in Q$ where $i \in B$ if and only if the edge
  $E(P_i) \cap C_{a-1}$ is incident with $v$). 
  In time linear in the mass of $V_a$, we can compute the connected
  components of $V_a$, and the corresponding partition $Q'$ of $\righti(W_a)$.
  Then, as long as the current sequence of blocks is not yet connected
  (i.e., as long as $Q' \neq \{[\lambda]\}$), repeat the process for every
  block $a' \geq a$:  Let $H=G[V_{a'}]$; for every block $B \in Q$,
  add a vertex $s_B$ to $H$, connected to the endpoints of $P_i$
  in $V_{a'}$ for every $i \in B$; and compute the connected components of $H$
  and the corresponding partition $Q''$ of $\righti(W_{a'})$.  
  Clearly this takes linear time and allows us to detect the first block $V_{b+1}$
  such that $V_a \cup \ldots \cup V_b$ is connected.  Then the next bundle
  contains blocks $V_a$ through $V_{b-1}$. 
\end{proof}

We can now prove Theorem~\ref{theorem:flow-augmentation}.
We refrain from optimizing the exponent of $k$ in the running time, since every plausible application of the theorem will have an overhead of $2^{\Oh(k \log k)}$ separate applications anyway.

\begin{proof}[Proof of Theorem~\ref{theorem:flow-augmentation}]
  To bound the running time, we make two notes.  First, in every call to one of the algorithms \algosample, \algoshort or \algoshortone, recursive calls are only made on disjoint vertex sets of the respective graph $G$ (excepting special vertices $s$, $t$).  Furthermore, on each call into \algosample we either have a decreased value of $k$ or an increased value of $\lambda_G(s,t)$, and there are only $O(k^2)$ possible combined values of $(k,\lambda)$. 
  Thus every vertex of $G$ except $s$, $t$ is processed in at most a polynomial number of process calls.

  Second, we note that the density of the graph $G+A$ does not increase too much beyond the density of $G$. Specifically, it is easy to verify that for every vertex $v \in V(G)$, at most $k^{\Oh(1)}$ new edges are added incident with $v$. 
  Hence it suffices that the local work in each procedure is linear-time in the size of the graph it is called on.
  For this, the only part that needs care is the computation of bundles, Lemma~\ref{lemma:fast-bundles}.
  Every other step is immediate. Hence the running time is bounded by some $k^{\Oh(1)}\Oh(m)$. 
  
The rest of the statement --- namely, the fact that in the output $(A,\witnessflow)$ 
we have $\lambda_{G+A}(s,t) \geq \lambda^\ast$, $\witnessflow$ is an $(s,t)$-flow of size $\lambda^\ast$,
and that with probability $2^{-\Oh(k \log k)}$ the pair $(A,\witnessflow)$ is compatible with $Z$,
for any eligible $(s,t)$-cut $Z$ of size at most $k$
--- now follows via joined induction from Lemma~\ref{lemma:outerloop:correct}, Lemma~\ref{lemma:short-one} and Lemma~\ref{lemma:short-correct}.
  \end{proof}

\bibliographystyle{alpha}
\bibliography{../references.bib}

\newcommand{\etalchar}[1]{$^{#1}$}
\begin{thebibliography}{KKPW21b}

\bibitem[BDT18]{BousquetDT18}
Nicolas Bousquet, Jean Daligault, and St{\'{e}}phan Thomass{\'{e}}.
\newblock Multicut is {FPT}.
\newblock {\em {SIAM} J. Comput.}, 47(1):166--207, 2018.

\bibitem[CCH{\etalchar{+}}16]{ChitnisCHPP16}
Rajesh Chitnis, Marek Cygan, MohammadTaghi Hajiaghayi, Marcin Pilipczuk, and
  Michal Pilipczuk.
\newblock Designing {FPT} algorithms for cut problems using randomized
  contractions.
\newblock {\em {SIAM} J. Comput.}, 45(4):1171--1229, 2016.

\bibitem[CEM17]{ChitnisEM17}
Rajesh Chitnis, L{\'{a}}szl{\'{o}} Egri, and D{\'{a}}niel Marx.
\newblock List {H}-coloring a graph by removing few vertices.
\newblock {\em Algorithmica}, 78(1):110--146, 2017.

\bibitem[CKL{\etalchar{+}}21]{CyganKLPPSW21}
Marek Cygan, Pawel Komosa, Daniel Lokshtanov, Marcin Pilipczuk, Michal
  Pilipczuk, Saket Saurabh, and Magnus Wahlstr{\"{o}}m.
\newblock Randomized contractions meet lean decompositions.
\newblock {\em {ACM} Trans. Algorithms}, 17(1):6:1--6:30, 2021.

\bibitem[CLP{\etalchar{+}}19]{CyganLPPS19}
Marek Cygan, Daniel Lokshtanov, Marcin Pilipczuk, Michal Pilipczuk, and Saket
  Saurabh.
\newblock Minimum bisection is fixed-parameter tractable.
\newblock {\em {SIAM} J. Comput.}, 48(2):417--450, 2019.

\bibitem[Die12]{Diestel_book}
Reinhard Diestel.
\newblock {\em Graph Theory, 4th Edition}, volume 173 of {\em Graduate texts in
  mathematics}.
\newblock Springer, 2012.

\bibitem[KKPW21a]{dfl-arxiv}
Eun~Jung Kim, Stefan Kratsch, Marcin Pilipczuk, and Magnus Wahlstr{\"{o}}m.
\newblock Directed flow-augmentation.
\newblock {\em CoRR}, abs/2111.03450, 2021.

\bibitem[KKPW21b]{ufl-soda}
Eun~Jung Kim, Stefan Kratsch, Marcin Pilipczuk, and Magnus Wahlstr{\"{o}}m.
\newblock Solving hard cut problems via flow-augmentation.
\newblock In D{\'{a}}niel Marx, editor, {\em Proceedings of the 2021 {ACM-SIAM}
  Symposium on Discrete Algorithms, {SODA} 2021, Virtual Conference, January 10
  - 13, 2021}, pages 149--168. {SIAM}, 2021.

\bibitem[KKPW22]{csp-arxiv}
Eun~Jung Kim, Stefan Kratsch, Marcin Pilipczuk, and Magnus Wahlstr{\"{o}}m.
\newblock Flow-augmentation {III:} complexity dichotomy for boolean csps
  parameterized by the number of unsatisfied constraints.
\newblock {\em CoRR}, abs/2207.07422, 2022.
\newblock To appear at SODA 2023.

\bibitem[KLM{\etalchar{+}}20]{KratschLMPW20}
Stefan Kratsch, Shaohua Li, D{\'{a}}niel Marx, Marcin Pilipczuk, and Magnus
  Wahlstr{\"{o}}m.
\newblock Multi-budgeted directed cuts.
\newblock {\em Algorithmica}, 82(8):2135--2155, 2020.

\bibitem[MR14]{MarxR14}
D{\'{a}}niel Marx and Igor Razgon.
\newblock Fixed-parameter tractability of multicut parameterized by the size of
  the cutset.
\newblock {\em {SIAM} J. Comput.}, 43(2):355--388, 2014.

\bibitem[PY01]{PapadimitriouY01}
Christos~H. Papadimitriou and Mihalis Yannakakis.
\newblock Multiobjective query optimization.
\newblock In {\em {PODS}}. {ACM}, 2001.

\end{thebibliography}

\appendix

\section{Proof of Claim~\ref{cl:fa-calculus}}\label{app:fa-calculus}

The following proof of the following lemma, being essentially
a restatement of Claim~\ref{cl:fa-calculus},
is due to Piotr Nayar. We thank Piotr for allowing us to include here the proof.

  \begin{lemma}\label{lem:fa-calculus}
For every $c_1 \geq 0$ there exist $c_2 \geq 0$ such that for all $x_1, x_2, x_0 >0$ we have
\begin{align*}
	& c_2(x_1 + x_2 + 2 x_0)  \ln(x_1 + x_2 + 2 x_0)  + x_1 \ln (x_1) + x_2 \ln (x_2) + 2 x_0  \ln(x_0) & \\
	& \qquad  \geq  
(x_1 + x_2 + 2 x_0) \ln (x_1 + x_2 + 2 x_0)  + c_2 (x_1 + x_0) \ln (x_1 + x_0) + c_2 (x_2 + x_0) \ln (x_2 + x_0) + c_1 x_0.
\end{align*}
  \end{lemma}

\begin{proof}
The inequality is homogeneous under $(x_1, x_2, x_0) \to (\alpha x_1, \alpha x_2, \alpha x_0)$. We can therefore assume that $x_1+x_2+2x_0=1$.
We then introduce $\lambda \in (0,1)$ such that  $x_1+x_0=\lambda$ and $x_2+x_0=1-\lambda$. The inequality is invariant under $(x_1, x_2) \to (x_2, x_1)$, so we can assume that $\lambda \in(0, 1/2]$. Our goal is now to prove that
\[
	 x_1 \ln (x_1) + x_2 \ln (x_2) + 2 x_0  \ln(x_0) - c_1 x_0 \geq  + c_2\left[ \lambda \ln \lambda + (1-\lambda) \ln(1-\lambda) \right].
\]  
The function $z \mapsto z \ln z$ is decreasing on $(1,1/e)$ and increasing on $(1/e,1)$.
We consider two cases. \\

\noindent \textit{Case 1.} $ \lambda \in ((e-1)/(2e),1/2]$. In this case the left hand side if bounded from below by $-\frac{4}{e \ln 2}-c_1$ and $ \lambda \log \lambda + (1-\lambda) \log(1-\lambda)$ is bounded from above by some negative constant, so there is nothing to prove. \\

\noindent \textit{Case 2.} $\lambda \leq \frac{e-1}{2e}<\frac{1}{e}$. In this case if $0<x \leq \lambda$, then $x \log x \geq \lambda \log \lambda$ and thus $x_1 \log (x_1) + 2x_0 \log(x_0) \geq 3 \lambda \log \lambda$. Moreover, $x_2 \geq 1-2\lambda$ and therefore $x_2 \log(x_2) \geq  (1-2\lambda)\log(1-2\lambda)$, since $1-2\lambda \geq \frac{1}{e}$. After applying theses bounds we have to show that
\[
	 3 \lambda \log \lambda + (1-2\lambda) \log(1-2\lambda) - c_1 \lambda \geq  c_2\left[ \lambda \log \lambda + (1-\lambda) \log(1-\lambda) \right].
\]
In other words, we want to show that there exists $c_2$ such that
\[
	\frac{3 \lambda \log \lambda + (1-2\lambda) \log(1-2\lambda) - c_1 \lambda}{  \lambda \log \lambda + (1-\lambda) \log(1-\lambda) } \leq c_2, \qquad 0<\lambda \leq \frac{e-1}{2e}.
\]
This can be verified by checking that the limit $\lambda \to 0^+$ is finite and thus it is enough to take $c_2$ to be the supremum of the left hand side.
\end{proof}
\end{document}